%% file: main.tex
\documentclass{article}

\usepackage{fullpage}

\usepackage[utf8]{inputenc}
\usepackage{amsmath,amsthm,amssymb}
\usepackage{mathtools}
\usepackage{bbm}
\usepackage{hyperref}
\usepackage{algorithm}
\usepackage{algorithmicx,algpseudocode}
\usepackage{thm-restate}

\usepackage{cleveref}

\usepackage[normalem]{ulem}

\crefformat{equation}{(#2#1#3)}

\newcommand{\RR}{\mathcal{R}}

\newcommand{\XX}{\mathcal{X}}
\newcommand{\MM}{\mathcal{M}}

\newcommand{\br}[1]{\left[#1\right]}
\newcommand{\pr}[1]{\left(#1\right)}
\newcommand{\norm}[1]{\left|\left|#1\right|\right|}
\newcommand{\norme}[1]{\norm{#1}_2}

\newcommand{\CC}{\mathcal{C}}
\newcommand{\FF}{\mathcal{F}}
\newcommand{\LL}{\mathcal{L}}

\newcommand{\ra}{\rightarrow}
\newcommand{\la}{\leftarrow}

\newcommand{\FixedM}{{\textsc{BBScheduler}}}

\newcommand{\ImprovingM}{\textsc{BBImprover}}
\newcommand{\SmallDB}{\textsc{SmallDB}}
\newcommand{\SmallDBG}{\textsc{SmallDBG}}

\DeclareMathOperator*{\PrSym}{Pr}
\newcommand{\Prob}[2]{\PrSym_{#1}\left[#2\right]}

\newcommand{\ones}{\mathbbm{1}}

\DeclareMathOperator{\Range}{Range}

\newcommand{\cut}[1]{}

\usepackage{xcolor}
\definecolor{DarkGreen}{rgb}{0.1,0.5,0.1}

\input{Header}

\newcommand{\X}{\mathcal{X}}

\newcommand{\RE}[2]{\operatorname*{RE}\left( #1 || #2 \right)}

\begin{document}

\title{Differential Privacy for Growing Databases}
\author{Rachel Cummings\thanks{School of Industrial and Systems Engineering, Georgia Institute of Technology. Email: {\tt rachelc@gatech.edu}. Supported in part by a Mozilla Research Grant.} \and Sara Krehbiel\thanks{Department of Math and Computer Science, University of Richmond. Email: \texttt{krehbiel@richmond.edu}. Supported in part by a Mozilla Research Grant.} \and Kevin A. Lai\thanks{School of Computer Science, Georgia Institute of Technology. Email: \texttt{kevinlai@gatech.edu}. Supported in part by National Science Foundation grant IIS-1453304.} \and Uthaipon (Tao) Tantipongpipat\thanks{School of Computer Science, Georgia Institute of Technology. Email: \texttt{uthaipon3@gatech.edu}.  Supported in part by National Science Foundation grants CCF-24067E5 and CCF-1740776 (Georgia Institute of Technology TRIAD), and by a Georgia Institute of Technology ARC fellowship.  Part of this work was completed while the author was a student at University of Richmond.}}
\maketitle

\begin{abstract}
We study the design of differentially private algorithms for adaptive analysis of dynamically growing databases, where a database accumulates new data entries while the analysis is ongoing.  We provide a collection of tools for machine learning and other types of data analysis that guarantee differential privacy and accuracy as the underlying databases grow arbitrarily large. We give both a general technique and a specific algorithm for adaptive analysis of dynamically growing databases. Our general technique is illustrated by two algorithms that schedule black box access to some algorithm that operates on a fixed database to generically transform private and accurate algorithms for static databases into private and accurate algorithms for dynamically growing databases. These results show that almost any private and accurate algorithm can be rerun at appropriate points of data growth with minimal loss of accuracy, even when data growth is unbounded. Our specific algorithm directly adapts the private multiplicative weights algorithm of \cite{HR10} to the dynamic setting, maintaining the accuracy guarantee of the static setting through unbounded data growth.  Along the way, we develop extensions of several other differentially private algorithms to the dynamic setting, which may be of independent interest for future work on the design of differentially private algorithms for growing databases.
\end{abstract}



\input{introNotation.tex}

\input{results}

\input{prelims.tex}

\input{OldKevin.tex}

\input{Tao.tex}
\bibliographystyle{alpha}
\bibliography{privacy}

\appendix

\input{smalldbproofs}


\input{pmwproofs}

\input{app}

\end{document}

%% file: Header.tex
\newcommand{\N}{\ensuremath{\mathbb{N}}}

\newcommand{\R}{\ensuremath{\mathbb{R}}}

\newcommand{\M}{\ensuremath{\mathcal{M}}}

\theoremstyle{definition}

\newtheorem{definition}{Definition}

\newtheorem{theorem}{Theorem}

\newtheorem{lemma}[theorem]{Lemma}
\newtheorem{corollary}[theorem]{Corollary}

\numberwithin{equation}{section}

\DeclareMathOperator{\poly}{poly}

\DeclareMathOperator*{\E}{E}

\newcommand{\abs}[1]{\lvert{#1}\rvert}

\newcommand{\set}[1]{\{{#1}\}}

\newcommand{\ceil}[1]{\lceil{#1}\rceil}

\newcommand{\eps}{\epsilon}
\newcommand{\Lap}{\text{Lap}}

\usepackage[normalem]{ulem}
\usepackage{color}



%% file: introNotation.tex
%


\section{Introduction}\label{s.intro}

Technological advances are continuously driving down the cost of data collection and storage. Data collection devices such as smartphones and wearable health monitors have become ubiquitous, resulting in continuous accumulation of new data. This means that the statistical properties of a database may evolve dramatically over time, and earlier analysis of a database may 
 grow stale.  For example, tasks like identifying trending news topics critically rely on dynamic data and dynamic analysis. 
To harness the value of growing databases and keep up with data analysis needs, guarantees of machine learning algorithms and other statistical tools must apply not just to fixed databases but also to dynamic databases.







Learning algorithms must deal with highly personal 
 data in contexts such as wearable health data, browsing behavior, and GPS location data. 
 In these settings, privacy concerns are particularly important.   
Analysis of sensitive data without formal privacy guarantees has 
 led to numerous privacy violations in practice, for tasks such as recommender systems \cite{CKN11}, targeted advertising \cite{Kor10}, data anonymization \cite{NS08}, and deep learning \cite{HAP17}. 




In the last decade, a growing literature on differential privacy has developed to address these concerns (see, e.g., \cite{DR14}).  First defined by \cite{DMNS}, differential privacy gives a mathematically rigorous worst-case bound on the maximum amount of information that can be learned about any one individual's data from the output of an algorithm.  The theoretical computer science community has been prolific in designing differentially private algorithms that provide accuracy guarantees for a wide variety of machine learning problems (see \cite{JLE14} for a survey).  Differentially private algorithms have also begun to be implemented in practice by major organizations such as Apple, Google, Uber, and the United Status Census Bureau.  However, the vast majority of differentially private algorithms are designed only for static databases, and are ill-equipped to handle new environments with growing data.




This paper presents a collection of tools for machine learning and other types of data analysis that guarantee differential privacy and accuracy as the underlying databases grow arbitrarily large. We give both a general technique and a specific algorithm for adaptive analysis of dynamically growing databases. Our general technique is illustrated by two algorithms that schedule black box access to some algorithm that operates on a fixed database, to generically transform private and accurate algorithms for static databases into private and accurate algorithms for dynamically growing databases. 
Our specific algorithm directly adapts the private multiplicative weights algorithm of~\cite{HR10} to the dynamic setting, maintaining the accuracy guarantee of the static setting through unbounded data growth.

%% file: results.tex

\subsection{Our Results}


Here we outline our two sets of results for adaptive analysis of dynamically growing databases. Throughout the paper, we refer to the setting in which a database is fixed for the life of the analysis as the {\em static setting}, and we refer to the setting in which a database is accumulating new data entries while the analysis is ongoing as the {\em dynamic setting}. 

Our first set of results consists of two methods for generically transforming a black box algorithm that is private and accurate in the static setting into an algorithm that is private and accurate in the dynamic setting. \FixedM{} reruns the black box algorithm every time the database increases in size 
 by a small multiplicative factor, and it provides privacy and accuracy guarantees that are independent of the total number of queries (Theorem \ref{thm.fixed}).  \FixedM{} calls each successive run of the black box algorithm with an exponentially shrinking privacy parameter to achieve any desired total privacy loss. The time-independent accuracy guarantee arises from the calibration of the decreasing per-run privacy parameters with the increasing database size. We instantiate this scheduler using the \SmallDB{} algorithm \cite{blr08} for answering $k$ linear queries on a database of size $n$ over a universe of size $N$. With our scheduler we can answer an infinite amount of queries from a linear query class of size $k$ on a growing database with starting size $n$ over universe of size $N$. The static and dynamic settings have the following respective accuracy guarantees (Theorem~\ref{thm:staticSmallDB}, \cite{blr08}; Theorem~\ref{thm:smallDBG}): 
 $$\textstyle \alpha_{\text{s}}=O\left(\frac{\log N \log k}{\eps n}\right)^{1/3} \qquad \alpha_{\text{d}} = O\left(\frac{\log N \log k}{\eps n}\right)^{1/5}$$ Our second transformation, \ImprovingM{}, runs the black box every time new entries are added to the database, and it yields accuracy guarantees that improve as more data accumulate.  This algorithm is well-suited for problems where data points are sampled from a distribution, where one would expect the accuracy guarantees of static analysis to improve with the size of the sample. We apply this scheduler to private empirical risk minimization (ERM) algorithms to output classifiers with generalization error that improves as the training database grows (Theorem \ref{thm.erm}).

The following informal theorem statement summarizes our results for \FixedM{} (Theorem \ref{thm.fixed}) 
and \ImprovingM{} (Theorem \ref{thm.improvealpha}).  These results
 show that almost any private and accurate algorithm can be rerun at appropriate points of data growth with minimal loss of accuracy.  Throughout the paper, we use $n$ to denote the starting size of the database. The $\tilde{O}$ below hides $\poly\log(n)$ terms, and we suppress dependence on parameters other than $\eps$ and $n$, (e.g., data universe size $N$, number of queries $k$, failure probability $\beta$). Section~\ref{s.kevin} also provides an improved $(\eps,\delta)$-private version of \FixedM{}. 


\begin{theorem}[Informal] 
Let $\MM$ be an $\epsilon$-differentially private algorithm that is $(\alpha,\beta)$-accurate for an input query stream for some $\alpha = \tilde{O}\left(\left(\frac{1}{\eps n}\right)^p\right)$ and constant $p$. Then  
\begin{enumerate}
\item \FixedM{} running $\MM$ is $\eps$-differentially private and $(\alpha,\beta)$-accurate for $\alpha = \tilde{O}\left(\left(\frac{1}{\eps n}\right)^{p/(2p+1)}\right)$.
\item \ImprovingM{} running $\MM$ is $(\eps,\delta)$-differentially private and is $(\{\alpha_t\}_{t\ge n},\beta)$-accurate for $\alpha_t = \tilde{O}\left(\left(\frac{\sqrt{\log(1/\delta)}}{\eps \sqrt{t}}\right)^p\right)$, where $\alpha_t$ bounds the error when the database is size $t\ge n$.
\end{enumerate}
\end{theorem}

Our second set of results opens the black box to increase accuracy and adaptivity by modifying the private multiplicative weights (PMW) algorithm~\cite{HR10}, a broadly useful algorithm for privately answering an adaptive stream of $k$ linear queries with accuracy $\alpha=O(\frac{\log N\log k}{\eps n})^{1/3}$.  Our modification for growing databases (PMWG) considers all available data when any query arrives, and it suffers asymptotically no additional accuracy cost relative to the static setting.  

The static PMW algorithm answers an adaptive stream of queries while maintaining a public histogram reflecting the current estimate of the database given all previously answered queries. It categorizes incoming queries as either easy or hard, suffering significant privacy loss only for the hard queries. Hardness is determined with respect to the public histogram: upon receipt of a query for which the histogram provides a significantly different answer than the true database, PMW classifies this as a hard query, and it updates the histogram in a way that moves it closer to a correct answer on that query. 
The number of hard queries is bounded using a potential argument. Potential is defined as the relative entropy between the database and the public histogram. This quantity is initially bounded, it decreases by a substantial amount after every hard query, and it never increases. 

The main challenge in adapting PMW to the dynamic setting is that we can no longer use this potential argument to bound the number of hard queries. This is because the relative entropy between the database and the public histogram can increase as new data arrive. 
 In the worst case, PMW can learn the database with high accuracy (using many hard queries), and then adversarial data growth can change the composition of the database dramatically, allowing for even more hard queries on the new data than is possible in the static setting. 
 Instead, we modify PMW so that the public histogram updates not only in response to hard queries but also in response to new data arrivals. By treating the new data as coming from a uniform distribution, these latter updates incur no additional privacy loss, and they mitigate the relative entropy increase due to new data. This modification allows us to maintain the accuracy guarantee of the static setting through unbounded data growth. The following informal theorem is a statement of our main result for PMWG (Theorem \ref{SPMW-final-thm}). As with the static PMW algorithm, we can improve the exponent in the bound to $1/2$ if our goal is $(\eps,\delta)$-privacy for $\delta>0$ (Theorem~\ref{SPMW-delta-final-thm}).

\begin{theorem}[Informal] PMWG is $\eps$-differentially private and $(\alpha,\beta)$-accurate for any stream of up to $\kappa\cdot \exp(\sqrt{t/n})$ queries when the database is any size $t\ge n$ for some $\alpha = {O}\left(\left(\frac{\log N \log \kappa}{\eps n}\right)^{1/3}\right)$.
\end{theorem}

Along the way, we develop extensions of several static differentially private algorithms to the dynamic setting.  These algorithms are presented in Appendix \ref{sec:backgroundalgs}, and may be of independent interest for future work on the design of differentially private algorithms for growing databases.

%

\subsection{Related Work}


\textbf{Online Learning.} Our setting of dynamically growing databases is most closely related to online learning, where a learner plays a game with an adversary over many rounds. On each round $t$, the adversary first gives the learner some input, then the learner chooses an action $a_t$ and receives loss function $\LL_t$ chosen by the adversary, and experiences loss $\LL_t(a_t)$.  There is a vast literature on online learning, including several works on differentially private online learning \cite{JKT12,smith13,agarwal17}. In those settings, a database is a sequence of loss functions, and neighboring databases differ on a single loss function. While online learning resembles the dynamic database setting, there are several key differences. Performance bounds in the online setting are in terms of regret, which is a cumulative error term. On the other hand, we seek additive error bounds that hold for all of our answers. Such bounds are not possible in general for online learning, since the inputs are adversarial and the true answer is not known. In our case, we can achieve such bounds because even though queries are presented adversarially, we have access to the query's true answer. Instead of a cumulative error bound, we manage a cumulative privacy budget. 

\textbf{Private Learning on a Static Database.} There is a prominent body of work designing differentially private algorithms in the static setting for a wide variety of machine learning problems (see \cite{JLE14} for a survey).  These private and accurate algorithms can be used as black boxes in our schedulers \FixedM{} and \ImprovingM{}. In this paper, we pay particular attention to the problem of private empirical risk minimization (ERM) as an instantiation for our algorithms.  Private ERM has been previously studied by \cite{CMS11, KST12, BST14}; we compare our accuracy bounds in the dynamic setting to their static bounds in \Cref{fig:erm}.\footnote{To get the static bounds, we use Appendix D of \cite{BST14}, which converts bounds on expected excess empirical risk to high probability bounds.}



\textbf{Private Adaptive Analysis of a Static Database.} If we wish to answer multiple queries on the same database by independently perturbing each answer, then the noise added to each answer must scale linearly with the number of queries to maintain privacy, meaning only $O(n)$ queries can be answered with meaningful privacy and accuracy. If the queries are known in advance, however, \cite{blr08} showed how to answer exponentially many queries relative to the database size for fixed $\eps$ and $\alpha$. Later, Private Multiplicative Weights (PMW) \cite{HR10} achieved a similar result in the interactive setting, where the analyst can adaptively decide which queries to ask based on previous outputs, with accuracy guarantees close to the sample error.  A recent line of work \cite{DFH+15, CLN+16, BNS+15} showed deep connections between differential privacy and adaptive data analysis of a static database.  Our results would allow analysts to apply these tools on dynamically growing databases.


\textbf{Private Non-Adaptive Analysis of a Dynamic Database.}  Differential privacy for growing databases has been studied for a limited class of problems.  Both \cite{DNPR10} and \cite{CSS11} adapted the notion of differential privacy to streaming environments in a setting where each entry in the database is a single bit, and bits arrive one per unit time.  \cite{DNPR10} and \cite{CSS11} design differentially private algorithms for an analyst to maintain an approximately accurate count of the number 1-bits seen thus far in the stream. This technique was later extended by \cite{smith13} to maintain private sums of real vectors arriving online in a stream.  We note that both of these settings correspond to only a single query repeatedly asked on a dynamic database, precluding meaningful adaptive analysis.  To contrast, we consider adaptive analysis of dynamically growing databases, allowing the analyst exponentially many predicate queries to choose from as the database grows.

%% file: prelims.tex
\section{Preliminaries}\label{s.prelims}

\newcommand{\cX}{\mathcal X}

All algorithms in this paper take as input databases over some fixed data universe $\cX$ of finite size $N$. Our algorithms and analyses represent a finite database $D\in\cX^n$ equivalently as a fractional histogram $x\in\Delta(\cX)\subseteq \R^N$, where $x^i$ is the fraction of the database of type $i\in [N]$. When we say a database $x\in\Delta(\cX)$ has size $n$, this means that for each $i\in[N]$ there exists some $n_i\in\N$ such that $x^i=n_i/n$. 

If an algorithm operates over a single fixed database, we refer to this as the static setting. For the dynamic setting, we define a \emph{database stream} to be a sequence of databases $X=\{x_t\}_{t\ge n}$ starting with a database $x_n$ of size $n$ at time $t=n$ and increasing by one data entry per time step so that $t$ always denotes both a time and the size of the database at that time. Our dynamic algorithms also take a parameter $n$, which denotes the starting size of the database.

We consider algorithms that answer real-valued queries $f:\R^N\to \R$ with particular focus on {\em linear queries}. A linear query assigns a weight to each entry depending on its type and averages these weights over the database. We can interpret a linear query as a vector $f\in[0,1]^N$ and write the answer to the query on database $x\in\Delta(\cX)$ as $\langle f, x\rangle$, $f(x)$, or $x(f)$, depending on context. For $f$ viewed as a vector, $f^i$ denotes the $i$th entry. We note that an important special case of linear queries are counting queries, which calculate the proportion of entries in a database satisfying some boolean predicate over $\cX$. 

Many of the algorithms we study allow queries to be chosen {\em adaptively}, i.e., the algorithm accepts a stream of queries $F=\{f_j\}_{j=1}^k$ where the choice of $f_{j+1}$ can depend on the previous $j$ queries and answers. For the dynamic setting, we doubly index a stream of queries as $F=\{f_{t,:}\}_{t\ge n}=\{\{f_{t,j}\}_{j=1}^{\ell_t}\}_{t\ge n}$ so that $t$ denotes the size of the database at the time $f_{t,j}$ is received, and $j=1,\dots,\ell_t$ indexes the queries received when the database is  size $t$. 

The algorithms studied produce outputs of various forms. To evaluate accuracy, we assume that an output $y$ of an algorithm for query class $\FF$ (possibly specified by an adaptively chosen query stream) can be interpreted as a function over $\FF$, i.e., we write $y(f)$ to denote the answer to $f\in\FF$ based on the mechanism's output. We seek to develop mechanisms that are accurate in the following sense.

\begin{definition}[Accuracy in the static setting] \label{def:util}
For $\alpha,\beta>0$, an algorithm $\mathcal{M}$ is \emph{$(\alpha, \beta)$-accurate}  for query class $\FF$ if for any input database $x\in \Delta(\cX)$, the algorithm outputs $y$ such that $\abs{f(x)-y(f)}\le \alpha$ for all $f\in\FF$ with probability at least $1-\beta$. 
\end{definition}

In the dynamic setting, accuracy must be with respect to the current database, and the bounds may be parametrized by time.

\begin{definition}[Accuracy in the dynamic setting] 
\label{decreasing-accuracy-def} For $\alpha_n,\alpha_{n+1},\dots >0$ and $\beta>0$, an algorithm $\MM$ is $(\{\alpha_t\}_{t\ge n},\beta)$-accurate for query stream $F=\{f_{t,:}\}_{t\ge n}$ if for any input database stream $X=\{x_t\}_{t\ge n}$, the algorithm outputs $y$ such that $\abs{f_{t,j}(x_t)-y(f_{t,j})}\le \alpha_t$ for all $f_{t,j}\in F$ with probability at least $1-\beta$.
\end{definition}

\subsection{Differential Privacy and Composition Theorems} 

Differential privacy in the static setting requires that an algorithm produce similar outputs on {\em neighboring databases} $x\sim x'$, which differ by a single entry. In the dynamic setting, differential privacy requires similar outputs on \emph{neighboring database streams} $X,X'$ that satisfy that for some $t\ge n$, $x_\tau=x'_\tau$ for $\tau=n,\dots,t-1$ and $x_\tau\sim x_\tau'$ for $\tau\geq t$.\footnote{Note that this definition is equivalent to the definition of neighboring streams in \cite{CPW16}.}  In the definition below, a pair of {\em neighboring inputs} refers to a pair of neighboring databases in the static setting or a pair of neighboring database streams in the dynamic setting.

%

\begin{definition}[Differential privacy \cite{DMNS}] \label{def:priv}
For $\eps,\delta>0$, an algorithm $\mathcal{M}$ is  \emph{$(\eps, \delta)$-differentially private} if for any pair of neighboring inputs $x,x'$ and any subset $S \subseteq \Range(\mathcal{M})$, 
$$ \Pr[\mathcal{M}(x) \in S] \leq e^\eps \cdot \Pr[\mathcal{M}(x') \in S] + \delta .$$
When $\delta=0$, we will say that $\mathcal{M}$ is $\eps$-differentially private.
\end{definition}

We note that in the dynamic setting, an element in $\Range(\mathcal M)$ is an entire (potentially infinite) transcript of outputs that may be produced by $\mathcal M$. 


Differential privacy is typically achieved by adding random noise that scales with the {\em sensitivity} of the computation being performed. The sensitivity of any real-valued query $f:\Delta(\cX)\to\R$ is the maximum change in the query's answer due to the change of a single entry in the database, denoted $\Delta_f=\max_{x\sim x'}\abs{f(x)-f(x')}$. We note that a linear query on a database of size $n$ has sensitivity $1/n$.





The following composition theorems quantify how the privacy guarantee degrades as additional computations are performed on a database.

\begin{theorem}[Basic composition \cite{DMNS}]\label{thm:basic-composition}
Let \(\mathcal{M}_i\) be an \(\epsilon_i\)-differentially private algorithm for all \(i\in [k]\). Then the composition \(\M\) defined as \(\M(x)=(\M_i(x))_{i=1}^k\) 
is $\epsilon$-differentially private for $\epsilon = \sum_{i=1}^k \epsilon_i$.
\end{theorem}

Theorem \ref{thm:basic-composition} is useful to combine many differentially private algorithms to still achieve \((\epsilon,0)\)-differential privacy. Assuming the privacy loss in each mechanism is the same, the privacy loss from composing \(k\) mechanisms scales with \(k\).  There is an advanced composition theorem due to \cite{DRV10} that improves the privacy loss to roughly \(\sqrt{k}\) by relaxing from \((\epsilon,0)\)-differential privacy to \((\epsilon,\delta)\)-differential privacy.  However, advanced composition does not extend cleanly to the case where each \(\M_i\) has different \(\epsilon_i\).  Instead we use a composition theorem based on concentrated differential privacy (CDP) of \cite{BS16}.  This gives us the flexibility to compose differentially private mechanisms with different \(\epsilon_i\) to achieve \((\eps,\delta)\)-differential privacy, where $\eps$ scales comparably to the $\sqrt{k}$ bound of advanced composition.
 
%
%

\begin{theorem}[CDP composition, Corollary of \cite{BS16}]
\label{advanced-composition-eps-value}
Let \(\mathcal{M}_i\) be a \(\epsilon_i\)-differentially private algorithm for all \(i\in[k]\). Then the composition of all \(\mathcal{M}_i\) is $(\epsilon,\delta)$-differentially private for $\epsilon = \frac{1}{2}(\sum_{i=1}^k \epsilon_i^2)+\sqrt{2(\sum_{i=1 }^k \epsilon_i^2)\log(1/\delta)}$.
In particular, for \(\delta \leq e^{-1}\) and \(\sum_{i=1}^k \epsilon_i^2 \leq 1\), we have $\epsilon \leq2\sqrt{(\sum_{i=1}^k \epsilon_i^2)\log(1/\delta)}$.
\end{theorem}


\begin{proof}
The statement follows from the three following propositions in \cite{BS16}:
\begin{enumerate}
\item 
A mechanism that is \(\eps\)-DP is \(\frac{1}{2}\epsilon^2\)-zCDP.
\item 
Composition of \(\rho_1\)-zCDP and \(\rho_2\)-zCDP is a (\(\rho_1+\rho_2\))-zCDP mechanism
\item 
A \(\rho\)-zCDP mechanism is \((\rho+2\sqrt{\rho\log(1/\delta)},\delta)\)-DP for any \(\delta >\ 0\).
\end{enumerate}
\end{proof}

 Theorem \ref{advanced-composition-eps-value} shows that composing \(\epsilon_i\)-differentially private algorithms results in \((\epsilon,\delta)\)-differential priacy, where the privacy \(\epsilon\) scales with the $\ell_2$-norm of the vector \((\eps_i)_{i=1}^k\) and \(\log^{1/2}(1/\delta)\).

\subsection{Empirical Risk Minimization}\label{s.ermprelim}





Empirical risk minimization (ERM) is one of the most fundamental tasks in machine learning. 
The task is to find a good classifier $\theta\in\mathcal C$ from a set of classifiers $\mathcal C$, given a database $x$ of size $n$ sampled from some distribution $P$ over $\mathcal X =\R^d$ and loss function $\LL:\mathcal C \times \Delta(\mathcal X) \to \R$. The loss $\LL$ of a classifier $\theta$ on a finite database $x_n=\{z_1,\dots,z_n\}$ with respect to some $L:\mathcal C\times\mathcal X\to\R$ is defined as
$\LL(\theta; x_n) = \frac{1}{n}\sum_{i=1}^n L(\theta; z_i)$.
Common choices for $L$ include $0-1$ loss, hinge loss, and squared loss.

We seek to find a $\hat{\theta}$ with small \emph{excess empirical risk}, defined as,
\begin{equation}\label{eq.risk}
\hat{\mathcal{R}}_n(\hat{\theta}) = \LL(\hat{\theta}; x_n) - \min_{\theta \in \CC} \LL(\theta; x_n).
\end{equation}
In convex ERM, we assume that $L(\cdot; x)$ is convex for all $x \in \XX$ and that $\CC$ is a convex set. We will also assume that $\XX\subseteq \R^p$. Convex ERM is convenient because finding a suitable $\hat{\theta}$ reduces to a convex optimization problem, for which there exist many fast algorithms. Some examples of ERM include finding a $d$-dimensional median and SVM.

ERM is useful due to its connections to the \emph{true risk}, also known as the generalization error, defined as
$\mathcal{R}(\theta) = \E_{x \sim P}\br{L(\theta;x)}$.
That is, the loss function will be low in expectation on a new data point sampled from $P$. We can also define the \emph{excess risk} of a classifier $\hat{\theta}$:
\begin{align*}
\textup{ExcessRisk}(\hat{\theta}) = \E_{x \sim P}\br{L(\hat{\theta};x)} - \min_{\theta \in \CC} \E_{x \sim P}\br{L(\theta;x)}.
\end{align*}

ERM finds classifiers with low excess empirical risk, which in turn often have low excess risk. The following theorem relates the two. For completeness, we first give some definitions relating to convex empirical risk minimization. 
A convex body $\CC$ is a set such that for all $x,y\in \CC$ and all $\lambda \in [0,1]$, $\lambda x + (1-\lambda)y \in \CC$.	
 A vector $v$ is a subgradient of a function $L$ at $x_0$ if for all $x\in \CC$, $L(x) - L(x_0) \ge \langle v,x- x_0\rangle$.	
A function $L:\CC\ra \R$ is $G$-Lipschitz if for all pairs $x,y\in \CC$, $\abs{L(x) - L(y)} \le G\norme{x - y}$.
 $L$ is $\Delta$-strongly convex on $\CC$ if for all $x \in \CC$ and all subgradients $z$ at $x$ and all $y \in \CC$, we have $L(y) \ge L(x) + \langle z, y - x\rangle + \frac{\Delta}{2} \norme{y - x}^2$.
$L$ is $B$-smooth on $\CC$ if for all $x\in \CC$, for all subgradients $z$ at $x$ and for all $y \in \CC$, we have $L(y) \le L(x) + \langle z, y- x\rangle + \frac{B}{2}\norme{y - x}^2$. 	
We denote the diameter of a convex set $\CC$ by $\norme{\CC} = \arg \max_{x,y\in \CC} \norme{x - y}$.

\begin{theorem}[\cite{shalev09}]
For $G$-Lipschitz and $\Delta$-strongly convex loss functions, with probability at least $1-\gamma$ over the randomness of sampling the data set $x_n$, the following holds:
\begin{align*}
\textup{ExcessRisk}(\hat{\theta}) \le \sqrt{\frac{2G^2}{\Delta}\hat{\mathcal{R}}_n(\hat{\theta})} + \frac{4G^2}{\gamma \Delta n}.
\end{align*}     
\end{theorem}
Moreover, we can generalize this result to any convex and Lipschitz loss function $L$ by defining a regularized version of $L$, called $\tilde{L}$, such that $\tilde{L}(\theta; x) = L(\theta;x) + \frac{\Delta}{2}\norme{\theta}^2$. Then $\tilde{L}$ is $(L + \norme{\CC})$-Lipschitz and $\Delta$-strongly convex. Also note that:
\begin{align*}
\textup{ExcessRisk}_{L}(\theta) \le \textup{ExcessRisk}_{\tilde{L}}(\theta) +\frac{\Delta}{2}\norme{\CC}^2.
\end{align*}
Thus, ERM finds classifiers with low true risk in these settings. The following result for differentially private static ERM is due to~\cite{BST14} and provides a baseline for our work in the dynamic setting. 

\begin{theorem}[Static ERM \cite{BST14}] There exists an algorithm ERM$(x,\mathcal C, \LL, \eps,\alpha,\beta,n)$ for $\mathcal X \subseteq \R^d$ that is $\eps$-differentially private  and $(\alpha,\beta)$-accurate for static ERM as long as  $\norm{\mathcal C}_2=1$, $\LL$ is 1-Lipschitz, and for sufficiently large constant $C$, 
$$\alpha\ge C \frac{d\log (1/\beta)\log\log(1/\beta)}{\eps n}.$$
\end{theorem}

\subsection{SmallDB}\label{s.smalldb}

The \SmallDB{} algorithm \cite{blr08} is a differentially private algorithm for generating synthetic databases. For any input database $x$ of size $n$, class $\FF$ of linear queries, and accuracy parameter $\alpha$, the algorithm samples a database $y$ of size $\frac{\log\abs{\FF}}{\alpha}$ with exponential bias towards databases that closely approximate $x$  on all the queries in $\FF$. The main strength of \SmallDB{} is its ability to accurately answer exponentially many linear queries while still preserving privacy, captured in the following guarantee. 

\begin{theorem}[Static \SmallDB{} \cite{blr08}]\label{thm:staticSmallDB}
        The algorithm \SmallDB{}($x, \FF, \epsilon,\alpha,\beta, n$) is $\epsilon$-differentially private, and it is $(\alpha,\beta)$-accurate for linear query class $\FF$ of size $\abs{\FF}=k$ 
        as long as for sufficiently large constant $C$,
        \begin{align*}
\alpha \ge C \pr{ \frac{\log N \log k + \log (1/\beta)}{\eps n} }^{1/3}.  \label{eq:smalldbalpha}
        \end{align*}
\end{theorem}

This bound on $\alpha$  shows that for a fixed accuracy goal, the privacy parameter can decrease proportionally as the size of the input database size increases.

\subsection{Private Multiplicative Weights}
The static private multiplicative weights (PMW) algorithm~\cite{HR10} answers an adaptive stream of linear queries while maintaining a public histogram $y$, which reflects the current estimate of the static database $x$ given all previously answered queries. Critical to the performance of the algorithm is that it 
uses the public histogram to categorize incoming queries as either easy or hard, 
and it updates the histogram after hard queries in a way that moves it closer to a correct answer on that query. 
The number of hard queries is bounded using a potential argument, where the potential function is defined as the relative entropy between the database and the public histogram, i.e., $\RE{x}{y}=\sum_{i\in[N]} x^i \log(x^i/y^i)$. This quantity is initially bounded, it decreases by a substantial amount after every hard query, and it never increases. The following guarantee illustrates that this technique allows for non-trivial accuracy for exponentially many adaptively chosen linear queries.\footnote{The bounds cited here are from the updated version in \url{http://mrtz.org/papers/HR10mult.pdf}} 

\begin{theorem}[Static PMW \cite{HR10}]\label{thm:static-PMW}
The algorithm PMW($x,F,\eps,\delta,\alpha,\beta,n$) is \((\epsilon,\delta)\)-differentially private, and it is \((\alpha,\beta)\)-accurate for $k$ adaptively chosen linear queries $F$ as long as for sufficiently large constant $C$ 
\begin{align*}
\alpha \ge \begin{cases}
C\left(\frac{\log N \log(k/\beta)}{\epsilon n}\right)^{1/3} & \text{ for } \delta=0 \\
C\left(\frac{\log^{1/2} N \log(k/\beta)\log(1/\delta)}{\epsilon n}\right)^{1/2} & \text{ for } \delta>0 \\
\end{cases}
\end{align*}\end{theorem} 
This result is nearly tight in that any \(\epsilon\)-differentially private algorithm that answers \(k\) adaptively chosen linear queries on a database of size \(n\) must have error \(\alpha \geq \Omega((\frac{\log(N/n)\log k }{\epsilon n})^{1/2})\)  \cite{HR10}. PMW runs in time linear in the data universe size \(N\). If the incoming data entries are drawn from a distribution that satisfies a mild smoothness condition, a compact representation of the data universe can significantly reduce the runtime~\cite{HR10}. The same idea applies to our modification of PMW for the dynamic setting presented in Section~\ref{s.tao}, but we only present the inefficient and fully general algorithm.

%% file: OldKevin.tex

\newcommand{\cF}{\mathcal {F}}
\section{Extending Accuracy Guarantees to Growing Databases}\label{s.kevin}


In this section, we give two schemes for answering a stream of queries on a growing database, given black box access to a differentially private algorithm for the static setting.  Our results extend the privacy and accuracy guarantees of these static algorithms to the dynamic setting, even when data growth is unbounded.  We also instantiate our results with important mechanisms for machine learning that are private in the static setting.

In Section \ref{s.fixed}, we provide an algorithm \FixedM{} for scheduling repeated runs of a static algorithm. \FixedM{} is differentially private and provides $(\alpha, \beta)$-accurate answers to all queries, for $\alpha$ that does not change as the database grows or as more queries are asked.  In Section \ref{s.improve}, we provide a second algorithm \ImprovingM{} that allows the accuracy guarantee to improve as more data accumulate.  This result is well-suited for problems where data points are sampled from a distribution, where one would expect the accuracy guarantees of static analysis to improve with the size of the sample.  This algorithm is differentially private and $(\{\alpha_t\}, \beta)$-accurate, where $\alpha_t$ is diminishing inverse polynomially in $t$ (i.e., approaching perfect accuracy as the database grows large).  We also instantiate our results with important mechanisms for machine learning that are private in the static setting.


For ease of presentation, we restrict our results to accuracy of real-valued queries, but the algorithms we propose could be applied to settings with more general notions of accuracy or to settings where the black box algorithm itself can change across time steps, adding to the adaptivity of this scheme.

\subsection{Fixed Accuracy as Data Accumulate}\label{s.fixed}

In this section, we give results for using a private and accurate algorithm for the static setting as a black box to solve the analogous problem in the dynamic setting. Our general purpose algorithm \FixedM{} treats a static algorithm as a black box endowed with privacy and accuracy guarantees, and it reruns the black box whenever the database grows by a small multiplicative factor.  For concreteness, we first show in Section \ref{s.smalldbgrow} how our results apply to the case of the well-known \SmallDB{} algorithm, described in Section \ref{s.smalldb}. Then in Section \ref{s.fixedgen}, we present the more general algorithm.

\subsubsection{Application: \SmallDB{} for Growing Databases}\label{s.smalldbgrow}


Before presenting our result in full generality, we instantiate it on \SmallDB{} for concreteness, and show how to extend \SmallDB{} to the dynamic setting. Recall from Section \ref{s.smalldb} that the static \SmallDB{} algorithm takes in a database $x$, a class of linear queries $\FF$, and privacy parameter $\eps$, and accuracy parameters $\alpha$, $\beta$.  The algorithm is $\eps$-differentially private and outputs a smaller database $y$ of size $\frac{\log|\FF |}{\alpha^2}$, from which all queries in $\FF$ can be answered with $(\alpha,\beta)$-accuracy.

In the dynamic setting, we receive a database stream $X=\set{x_t}_{t\ge n}$, a stream of queries $F=\set{\set{f_{t,j}}_{j=1}^{\ell_t}}_{t\ge n}$ from some class of linear queries $\FF$, parameters $\eps, \alpha, \beta$, and starting database size $n$.  We still require $\eps$-differential privacy and $(\alpha,\beta)$-accuracy on the entire stream of queries, for $\alpha$ that remains fixed as the database grows.
 


We design the \SmallDBG{} algorithm that works by running \SmallDB{} at times $\{t_i\}_{i=0}^\infty$, where $t_i = (1+\gamma)^i n$ for some $\gamma < 1$ chosen by the algorithm.\footnote{For simplicity, we will assume that $(1+\gamma)^i n$ is integral for for all $i$. We can replace $(1+\gamma)^i n$ with $\ceil{(1+\gamma)^i n}$ and achieve the same bounds up to a small sub-constant additive factor.} 
We will label the time interval from $t_i$ to $t_{i+1} - 1$ as the $i^{th}$ \emph{epoch}. At the start of the $i^{th}$ epoch, we call \SmallDB{} on the current database with privacy parameter $\eps_i$, and output a synthetic database $y_i$ that will be used to answer queries received during epoch $i$.\footnote{Note that \SmallDBG{} will still give similar guarantees if the query class changes over time, provided that the black box \SmallDB{} at time $t_i$ uses the correct query class $\FF_i$ for times $t_i$ to $t_{i+1} - 1$. We could think of this as \SmallDBG{} receiving a \SmallDB{}($\cdot,\FF_i,\cdot,\cdot,\cdot$) as its black box in epoch $i$.} \SmallDBG{} provides the following guarantee:


\begin{theorem}\label{thm:smallDBG}
        \SmallDBG{}($X, F, \FF, \eps,\alpha,\beta,n$) is $\eps$-differentially private and can answer all queries in query stream $F$ from query class $\FF$ of size $\abs{\FF}=k$ with $(\alpha,\beta)$-accuracy\footnote{With a more careful analysis, one can show that the numerator in this accuracy bound can be taken to be $16\log N \log |\FF| + 4\log \frac{1}{\beta}$ to match the form of the bound in \Cref{thm:staticSmallDB}.}
         for sufficiently large constant $C$ and 
        \begin{align*}
        \alpha \geq C \pr{ \frac{\log N \log k\log ({1}/{\beta})}{\eps n}}^{{1}/{5}}.
        \end{align*}
\end{theorem}



\noindent Note that there is no bound on the number of queries or on the database growth.  The algorithm can provide answers to an arbitrary number of linear queries at any time.


There are two key technical properties that allow this result to hold. First, each data point added to a database of size $t$ can only change a linear query by roughly $1/t$. Thus, using synthetic database $y_i$ from time $t_i$ for queries before time $(1+\gamma)t_i$ will incur extra additive error of at most $\gamma$. Second, since the $t_i$'s grow by a multiplicative factor each time, the epochs become exponentially far apart and the total privacy loss (due to composition of multiple calls of \SmallDB{}) is not too large.  





\subsubsection{A General Black Box Scheduler}\label{s.fixedgen}



The results for \SmallDBG{} are an instantiation of a more general result that extends the privacy and accuracy guarantees of any static algorithm to the dynamic setting.  Our general purpose algorithm \FixedM{} treats a static algorithm as a black box endowed with privacy and accuracy guarantees, and reruns the black box whenever the database grows by a factor of $(1+\gamma)$.  Due to the generality of this approach, \FixedM{} can be applied to any algorithm that satisfies $\eps$-differential privacy and $(\alpha,\beta)$-accuracy, as specified in Definition \ref{def.blackbox}.

\begin{definition}[$(p,g)$-black box]\label{def.blackbox} 
An algorithm $\MM(x, \eps,\alpha,\beta, n)$ is a $(p,g)$-black box for a class of linear queries $\FF$ if it is $(\eps,0)$-differentially private and with probability $1-\beta$ it outputs $y:\FF\to\R$ such that $\abs{y(f)-x_n(f)}\le \alpha$ for every $f\in\FF$ when $\alpha\ge g\left(\frac{\log (1/\beta)}{\eps n}\right)^p$ for some $g$ that is independent of $\eps,n,\beta$.
\end{definition}


The parameter $g$ is intended to capture dependence on domain-specific parameters that affect the accuracy guarantee.  For example, \SmallDB{} is a $(1/3,(64\log N\log k)^{1/3})$-black box for an arbitrary set $\FF$ of $k$ linear queries, and its output $y$ is a synthetic database of size $\log k/\alpha^2$.  





Our generic algorithm \FixedM{} (Algorithm \ref{algo.scheduler}) will run the black box $\MM$ at times $\{t_i\}_{i=0}^\infty$ for $t_i = (1+\gamma)^i n$ with $\gamma < 1$ that depends on $p$ and $g$. The $i^{th}$ $\MM$ call will have parameters $\eps_i$ and $\beta_i$, and will use $y_i$ to answer queries received during the $i^{th}$ epoch, from $t_i$ to $t_{i+1} - 1$.  

\begin{algorithm}[h]
        \caption{\FixedM{}($X, F, \MM, \eps, \delta, \beta, n, p, g$)}
        \begin{algorithmic}
                \If {$\delta = 0$} \Comment Set growth between epochs
                \State Let $\gamma \la g^{\frac{1}{2p+1}} \pr{\frac{\log \frac{1}{\beta}}{\eps n}}^{\frac{p}{2p+1}}$
                \Else { $\delta > 0$}
                \State Let $\gamma \la g^{\frac{1}{1.5p+1}} \pr{\frac{\log \frac{1}{\beta}}{\eps n}}^{\frac{p}{1.5p+1}}$
                \EndIf
                
                \State Let $i \la -1$ 
                \For {$t \la n,n+1,...$}

                \If {$t = (1+\gamma)^{i+1} n$} \Comment Begin new epoch when database has grown sufficiently
                \State $i \la i+1$ 
        \If {$\delta =0$}
                \State Let $\eps_i \la \frac{\gamma^2 (i+1)}{(1+\gamma)^{i+2}}\eps$
        \Else { $\delta >0$}
        \State Let $\eps_i \la \frac{\gamma^{1.5}(i+1)}{(1+\gamma)^{i+1.5}}\frac{\eps}{3\sqrt{\log (1/\delta)}}$
        \EndIf
                \State Let $\beta_i \la \pr{\frac{\beta}{1+\beta}}^{i+1}$
	            \State Let $\alpha_i \la g\pr{\frac{\log \frac{1}{\beta_i}}{\epsilon_i (1+\gamma)^i n}}^p$
	            \State Let $y_i \la \MM\pr{x_t, \eps_i,\alpha_i , \beta_i}$ \Comment Rerun $\MM$ for new epoch on new parameters
                \EndIf
                \For {$j\la 1,...,\ell_t$}
                \State Output $y_i(f_{t,j})$ \Comment Answer queries at time $t$ with the output of $\MM$ from current epoch
                \EndFor
                \EndFor
                
        \end{algorithmic}\label{algo.scheduler}
\end{algorithm}

We now state our main result for \FixedM{}:

\begin{restatable}{theorem}{fixed}\label{thm.fixed}
Let $\MM$ be a $(p,g)$-black box for query class $\FF$. Then for any database stream $X$ and stream of linear queries $F$ over $\FF$, \FixedM{}($X, F, \MM, \eps, \delta, \beta, n, p,g$) is $(\eps ,\delta)$-differentially private for $\eps < 1$ and ($\alpha, \beta$)-accurate for sufficiently large constant $C$ and 
\begin{align*}
\alpha&\ge \begin{cases} 
C g^{\frac{1}{2p+1}} \pr{\frac{\log ({1}/{\beta})}{\eps n}}^{\frac{p}{2p+1}} & \text{ if } \delta = 0 \\
C g^{\frac{1}{1.5p+1}}\pr{\frac{\sqrt{\log ({1}/{\delta})}\log ({1}/{\beta})}{\eps n}}^{\frac{p}{1.5p+1}} & \text{ if } \delta >0
\end{cases}.
\end{align*}

\end{restatable}

Note that this algorithm can provide two different types accuracy bounds. If we desire $(\eps,0)$-differential privacy, then the accuracy bounds are slightly weaker, while if we allow $(\eps,\delta)$-differential privacy, we can get improved accuracy bounds at the cost of a small loss in privacy.  The only differences are how the algorithm sets $\gamma$ and $\eps_i$.  For a complete proof of Theorem \ref{thm.fixed}, see \Cref{s.bbapp}. We present a proof sketch below.

\begin{proof}[Proof sketch of Theorem \ref{thm.fixed}]

$\FixedM{}$ inherits its privacy guarantee from the black box $\MM$ and the composition properties of differential privacy.  When $\delta=0$, we use Theorem \ref{thm:basic-composition} (Basic Composition).  When $\delta>0$,  we use Theorem \ref{advanced-composition-eps-value} (CDP Composition).  These two cases require different settings of $\gamma$ and $\eps_i$ for their respective composition theorems to yield the desired privacy guarantee.

To prove the accuracy of \FixedM{} we require the following lemma, which bounds the additive error introduced by answering queries that arrive mid-epoch using the slightly outdated database from the end of the previous epoch.

\begin{restatable}{lemma}{additiveerr}\label{lem:additiveerr}
        For any linear query $f$ and databases $x_t$ and $x_\tau$ from a database stream $X$, where $\tau \in [t,(1+\gamma)t]$ for some $\gamma \in (0,1)$, 
        \begin{align*}
        \abs{x_\tau(f) - x_t(f)} \le \frac{\gamma}{1+\gamma}.
        \end{align*}
\end{restatable}

We combine this lemma with the algorithm's choice of $\gamma_i$ to show that with probability at least $1-\beta_i$, all mid-epoch queries are answered $\alpha$-accurately with respect to the current database.  The final step is to bound the overall failure probability of the algorithm.  Taking a union bound over the failure probabilities in each epoch, we complete the proof by showing that $\sum_{i=0}^{\infty} \beta_i \le \beta$.



%

\end{proof}

\subsection{Improving Accuracy as Data Accumulate}\label{s.improve}

In the previous section, our accuracy bounds stayed fixed as the database size increased.  However, in some applications it is more natural for accuracy bounds to improve as the database grows.  For instance, in empirical risk minimization (defined in Section \ref{s.ermprelim}) the database can be thought of as a set of training examples.  As the database grows, we expect to be able to find classifiers with shrinking empirical risk, which implies shrinking generalization error. More generally, when database entries are random samples from a distribution, one would expect accuracy of analysis to improve with more samples.



In this section, we extend our black box scheduler framework to allow for accuracy guarantees that improve as data accumulate.  Accuracy improvements over \FixedM{} are typically only seen once the database is sufficiently large.  We first instantiate our result for empirical risk minimization in Section \ref{s.growerm}, and then present the general result in Section \ref{s.genimprove}.

\subsubsection{Application: Empirical Risk Minimization for Growing Databases}\label{s.growerm}


In the static setting, an algorithm for empirical risk minimization (ERM) takes in a database $x$ of size $n$, and outputs a classifier from some set $\CC$ that minimizes a loss function $\LL$ on the sample data.  Increasing the size of the training sample will improve accuracy of the classifier, as measured by excess empirical risk (Equation \eqref{eq.risk}).  Given the importance of ERM, it is no surprise that a number of previous works have considered differentially private ERM in the static setting \cite{CMS11, KST12, BST14}.


For ERM in the dynamic setting, we want a classifier $y_t$ at every time $t \geq n$ that achieves low empirical risk on the current database, and we want the empirical risk of our classifiers to improve over time, as in the static case.  Note that the dynamic variant of the problem is strictly harder because we must produce classifiers at every time step, rather than waiting for sufficiently many new samples to arrive.  Releasing classifiers at every time step degrades privacy, and thus requires more noise to be added to preserve the same overall privacy guarantee.  Nonetheless, we will compare our private growing algorithm, which simultaneously provides accuracy bounds for every time step from $n$ to infinity, to private static algorithms, which are only run once.  




In ERMG, our algorithm for ERM in the dynamic setting, the sole query of interest is the loss function $\LL$ evaluated on the current database.  At each time $t$, ERMG receives a single query $f_t$, where $f_t$ evaluated on the database is $x_t(f_t)=\min_{\theta\in\CC} \LL(\theta;x_t)$.  The black box outputs $y_t$, which is a classifier from $\CC$ that can be used to evaluate the single query $y_t(f_t)=\LL(y_t;x_t)$.  Our accuracy guarantee at time $t$ is the difference between $y_t(f_t)$ and $x_t(f_t)$:
\begin{align*}
\alpha_t = \LL(y_t;x_t) - \min_{\theta\in\CC} \LL(\theta;x_t).
\end{align*}
This expression is identical to the excess empirical risk $\hat{\RR}_t(y_t)$ defined in Equation \eqref{eq.risk}.  Thus accurate answers to queries are equivalent to minimizing empirical risk. Our accuracy bounds are stated in \Cref{thm.erm}.

\begin{restatable}{theorem}{ermthm}\label{thm.erm}
Let $c>0$, and $\LL$ be a convex loss function that is 1-Lipschitz over some set $\CC$ with $\norme{\CC} = 1$. Then for any stream of databases $X$ with points in $\mathbb{R}^d$, ERMG($X,\LL,\CC,\eps,\delta, \beta,n$) is $(\eps, \delta)$-differentially private and with probability at least $1-\beta$ produces classifiers $y_t$ for all $t \geq n$ that for sufficiently large constant $C$ have excess empirical risk bounded by
\begin{align*}
\hat{\RR}_t(y_t) \le C \frac{d \log ({1}/{\beta})\sqrt{\log (1/\delta)}}{\sqrt{c}\eps t^{\frac{1}{2} - c}}.
\end{align*}
If $\LL$ is also $\Delta$-strongly convex,
\begin{align*}
\hat{\RR}_t(y_t) \le C \frac{d^2 \log^2 ({1}/{\beta})\log (1/\delta) }{\sqrt{c}\Delta \eps^2 t^{1-c}}.
\end{align*}
\end{restatable} 


The results in Theorem~\ref{thm.erm} all come from instantiating (the more general) \Cref{thm.improvealpha} stated in the next section, and the proof is in \Cref{s.appimprove}. We use the static $(\eps,0)$-differentially private algorithms of \cite{BST14} as black boxes.  The differing assumptions on $\LL$ allow us to use different ($p,g$)-black boxes with different input parameters in each case.  We compare our growing bounds to these static bounds in \Cref{fig:erm}.\footnote{To get the static bounds, we use Appendix D of \cite{BST14}, which converts bounds on expected excess empirical risk to high probability bounds.} Since ERMG provides $(\eps,\delta)$-differential privacy, we also include static $(\eps,\delta)$-differential privacy bounds for comparison in \Cref{fig:erm}.  The static bounds are optimal in $d,t,$ and $\epsilon$ up to log factors. 



\renewcommand{\arraystretch}{1.75}
\begin{table}[h]
        \begin{tabular}{p{2.4cm} | c | c | c}
                Assumptions & Static $(\epsilon,0)$-DP \cite{BST14} & Static $(\epsilon,\delta)$-DP \cite{BST14} & Dynamic $(\epsilon,\delta)$-DP (our results) \\ \hline
                1-Lipschitz and $\norme{C} = 1$ & $\frac{d\log \frac{1}{\beta}}{\eps t}$ & $\frac{\sqrt{d} \log^2(t/\delta)\log \frac{1}{\beta}}{\eps t}$ & $\frac{d \sqrt{\log (1/\delta)}\log \frac{1}{\beta}}{\sqrt{c}\eps t^{\frac{1}{2} - c}}$ \\ \hline
                ... and $\Delta$-strongly convex (implies $\Delta \le 2$) & $\frac{d^2 (\log t)\log^2 \frac{1}{\beta}}{\Delta \eps^2 t^2}$ & $ \frac{d\log^3(t/\delta)\log^2 \frac{1}{\beta} }{\Delta  \epsilon^2 t^2}$ & $\frac{d^2 \log (1/\delta) \log^2 \frac{1}{\beta}}{\sqrt{c}\Delta \eps^2 t^{1-c}}$  
        \end{tabular}
        \caption{Comparison of excess empirical risk upper bounds in the static case versus the dynamic case for a database of size $t$ under differing assumptions on the loss function $\LL$. Database entries are sampled from $\R^d$, and $c$ is any positive constant. We ignore leading multiplicative constants and factors of $\log \log \frac{1}{\beta}$ in the static bounds.  As in \cite{BST14}, we assume $\delta < 1/t$ for simplicity.\label{fig:erm} }
\end{table}

Note that the bounds we get for the growing setting have the same dependence on $\epsilon, \beta$, and $\Delta$ and better dependence on $\delta$. The dependence on $t$ in our bound is roughly the square root of that in the static bounds.  Compared to the static $(\epsilon,0)$-differential privacy bounds, our dependence on $d$ is the same, while the dependence is squared relative to the static $(\epsilon,\delta)$-differential privacy bounds.  

Given that the growing setting is strictly harder than the static setting, it is somewhat surprising that we have no loss in most of the parameters, and only minimal loss in the size of the database $t$. Thus, for ERM, performance in the static setting largely carries over to the growing setting.

\subsubsection{A General Black Box Scheduler for Improving Accuracy}\label{s.genimprove}

In this section we describe the general \ImprovingM{} algorithm, which achieves accuracy guarantees in the dynamic setting that improve as the database size grows.  The algorithm takes in a private and accurate static black box $\MM$, which it re-runs on the current database at every time step.  We require the following more general definition of black box to state the privacy and accuracy guarantees of \ImprovingM{}.

\begin{definition}[Definition of $(p,p',p'',g)$-black box] 
An algorithm $\MM(x,\eps,\alpha,\beta,n)$ is a $(p,p',p'',g)$-black box for a class of linear queries $\FF$ if it is $(\eps,0)$-differentially private and with probability $1-\beta$ it outputs some $y:\FF\to\R$ such that $\abs{y(f)-x_n(f)}\le \alpha$ for every $f\in\FF$ when $\alpha\ge g\left(\frac{1}{\eps n}\right)^p\log^{p''} n\log^ {p'}(1/\beta)$ for some $g$ that is independent of $\eps,n,\beta$.
\end{definition}

The algorithm \ImprovingM{} (Algorithm \ref{algo.improver}) will run the black box $\MM$ after each new data point arrives, starting at time $t = n$, using time-dependent parameters $\eps_t, \alpha_t, \beta_t$.  The output $y_t$ will be used to answer all queries that arrive at time $t$.  

\begin{algorithm}[h!]
        \caption{\ImprovingM{}($X,F, \MM, \eps, \delta, \alpha, \beta, n,p, p', p'', g, c$) }
        \begin{algorithmic}
                \For {$t \la n,n+1,...$}
                \State Let $\eps_t \la \frac{\sqrt{c}}{3\sqrt{\log(1/\delta)}}\frac{\eps}{t^{\frac{1}{2} + c}}$
                \State Let $\beta_t \la \frac{\beta}{2t^2}$
                \State Let $\alpha_t \la g\pr{\frac{1}{\epsilon_i t}}^p \log^{p''} n \log^{p'} \frac{1}{\beta}$
                \State Let $y_t \la \MM(x_t, \eps_t,\alpha_t, \beta_t)$
                \For {$j \la 1,...,\ell_t$}
                \State Output $y_t(f_{t,j})$
                \EndFor
                \EndFor
                
        \end{algorithmic}\label{algo.improver}
\end{algorithm}

The following theorem is our main result for \ImprovingM{}, which states that the algorithm is differentially private and $(\alpha_t, \beta)$-accurate for $\alpha_t$ that decreases inverse polynomially in $t$.  The complete proof is given in Appendix \ref{s.bbapp}.

\begin{restatable}{theorem}{improve}\label{thm.improvealpha}  
Let $c> 0$ and let $\MM$ be a $(p,p',p'',g)$-black box for query class $\FF$. Then for any database stream $X$ and stream of linear queries $F$ over $\FF$, \ImprovingM{}($X,F, \MM, \eps, \delta, \beta, n, p,p',p'',g,c$) is $(\eps,\delta)$-differentially private for $\eps < 1$ and $(\{\alpha_t\}_{t\geq n},\beta)$-accurate for sufficiently large constant $C$ and 
\begin{align*}
\alpha_t = Cg\log^{p'} ({1}/{\beta})\pr{\frac{\sqrt{\log (1/\delta)}}{\sqrt{c}\eps t^{\frac{1}{2} - 2c}}}^p .
\end{align*}

\end{restatable}
The free parameter $c$ in Theorem \ref{thm.improvealpha} can be any positive constant, and should be set to an arbitrarily small constant for the algorithm to achieve the best asymptotic performance.

\ImprovingM{} does not incur accuracy loss from ignoring new data points mid-epoch as in \FixedM{} because it runs $\MM$ at every time step.  However, this also means that privacy loss will accumulate much faster than in \FixedM{} because more computations are being composed.  To combat this and achieve overall privacy loss $\eps$, each run of $\MM$ will have increasingly strict (i.e., smaller) privacy parameter $\eps_t$.  The additional noise needed to preserve privacy will overpower the improvements in accuracy until the database grows sufficiently large, when the accuracy of \ImprovingM{} will surpass the comparable fixed accuracy guarantee of \FixedM{}.  For any $p>0$, the guarantees of \ImprovingM{} are stronger when $t \gg n^2$.  This suggests that an analyst's choice of algorithm should depend on her starting database size and expectations of data growth.



%% file: Tao.tex
%
\section{Private Multiplicative Weights for Growing Databases}\label{s.tao}

In this section, we show how to modify the private multiplicative weights (PMW) algorithm for adaptive linear queries~\cite{HR10} to handle continuous data growth. The first black box process \FixedM{} in the previous section shows that any algorithm can be rerun with appropriate privacy parameters at appropriate points of data growth with minimal loss of accuracy with respect to the intra-epoch data. However, in some settings it may be undesirable to ignore new data for long periods of time, even if the overall accuracy loss is small. Although \ImprovingM{} runs the black box algorithm at every step for eventual tighter accuracy bounds, these bounds are inferior until the database grows substantially. We now show how to open the black box and apply these scheduling techniques with a modification of PMW that considers all available data when a query arrives, achieving tight bounds on accuracy as soon as analysis begins and continuing through infinite data growth.

The static PMW algorithm answers an adaptive stream of queries while maintaining a public histogram reflecting the current estimate of the database given all previously answered queries. Critical to the performance of the algorithm is that it categorizes incoming queries as either easy or hard, suffering significant privacy loss only for the hard queries. Hardness is determined with respect to the public histogram: upon receipt of a query for which the histogram provides a significantly different answer than the true database, PMW classifies this as a hard query, and it updates the histogram in a way that moves it closer to a correct answer on that query. 
The number of hard queries is bounded using a potential argument. Potential is defined as the relative entropy between the database and the public histogram. This quantity is initially bounded, decreases by a substantial amount after every hard query, and never increases. 

If we run static PMW on a growing database, the previous potential argument fails because the relative entropy between the database and the public histogram can increase as new data arrive. In the worst case, PMW can learn the database with high accuracy (using many hard queries), and then adversarial data growth can change the composition of the database dramatically, increasing the number of possible hard queries well beyond the bound for the static case.  
Instead, we modify PMW so that the public histogram updates not only in response to hard queries but also in response to new data arrivals. By treating the new data as coming from a uniform distribution, these latter updates incur no additional privacy loss, and they mitigate the relative entropy increase due to new data. In fact, this modification allows us to suffer only constant loss in accuracy per query relative to the static setting, while maintaining this accuracy through unbounded data growth and accumulating additional query budget during growth.

\subsection{$(\eps,0)$-Differentially Private PMWG}\label{s.eps0}


Our formal algorithm for PMW for growing databases (PMWG) is given as Algorithm~\ref{SPMG-alg} below. We give an overview here to motivate our main results. The algorithm takes as input a database stream $X=\{x_t\}_{t\ge n}$ and an adaptively chosen query stream $F=\{\{f_{t,j}\}_{j=1}^{\ell_t}\}_{t\ge n}$. It also accepts privacy and accuracy parameters $\eps,\delta,\alpha$.  In this section we restrict to the case where $\delta=0$; in Section \ref{s.epsdelta}, we allow $\delta>0$.

\begin{algorithm} 
\caption{\textsc{PMWG}($X,F,\eps,\delta,\alpha,n$) 
}\label{SPMG-alg}
\begin{algorithmic}
\If {$\delta=0$} 
\State Let $\xi_t \la \frac{\alpha^2 n^{1/2}}{162\log(Nn)}\epsilon t^{1/2}$ for $t\ge n$
\Else
\State Let $\xi_t \la \frac{\alpha n^{1/2}}{48\log^{1/2} (Nn)\log^{1/2}(1/\delta)}\epsilon t^{1/2}$ for $t\ge n$
\EndIf
\State Start NSG$(X,\cdot,2\alpha/3,\{\xi_t\}_{t\ge n})$ \Comment {Initialize subroutine}
\State Let $y_{n,0}^i \la1/N$ for  $i \in [N]$ \Comment{Public histogram}
\State Let $h_t \la 0$ for  $t\ge n$ \Comment{Hard query counters}
\State Let $b_t \la \frac{\log N}{t}+\frac{\log(t-1)}{t}+\log\frac{t}{t-1}$ for $t\ge n+1$ \Comment{Hard query bounds}

\For{each incoming query \(f_{t,j}\)}
\If {last query was at time $t'<t$}
\State  Let $y_{t,0}^i \la \frac{t'}{t}y_{t',\ell_{t'}}^i+\frac{t-t'}{t}\frac{1}{N}$ for $i\in[N]$ \Comment {Uniform update}
  \EndIf
  \State Let ${f'}_{t,2j-1} \la f_{t,j}-f_{t,j}(y_{t,j-1}), {f'}_{t,2j} \la f_{t,j}(y_{t,j-1})-f_{t,j}$
  \State Receive $a'_{t,2j-i},a'_{t,2j}$ from NSG on ${f'}_{t,2j-1}, {f'}_{t,2j}$ \Comment {Check hardness}
  \If {$a'_{t,2j-1}=\perp$ and $a'_{t,2j}=\perp$ }
    \State Let \(y_{t,j} \la  y_{t,j-1}\)
    \State Output $a_{t,j}\la f_{t,j}(y_{t,j})$ \Comment {Compute easy query answer}
  \Else
  \State Let \(h_t \la h_t+1\) 
  \If {$\sum_{\tau = n}^t h_{\tau} > \frac{36}{\alpha^2}\left(\log N+ \sum_{\tau=n+1}^t b_\tau\right)$}          
  \State \Return $\bot$ \Comment{Hard query budget exceeded}
  \EndIf
    \If {$a'_{t,2j-1} \in \R$}
      \State Output $a_{t,j} \la f_{t,j}(y_{t,j-1})+a'_{t,2j-1}$
    \Else 
      \State Output $a_{t,j} \la f_{t,j}(y_{t,j-1})-a'_{t,2j}$
    \EndIf \Comment {Compute hard query answer}
    
    \If {$a_{t,j}<f_{t,j}(y_{t,j-1})$}
     \State Let $r_{t,j} \la f_{t,j}$
    \Else
     \State Let $r_{t,j} \la 1-f_{t,j}$
    \EndIf

      \State Let $\hat{y}^i_{t,j} \la \exp\left(-\frac{\alpha}{6}r_{t,j}^i\right)y^i_{t,j-1}$ for $i\in[N]$
      \State Let $y^i_{t,j} \la \frac{\hat{y}^i_{t,j}}{\sum_{i'\in[N]}\hat{y}^{i'}_{t,j}}$ for $i\in[N]$ \Comment {MW update}
  
  \EndIf
\EndFor

\end{algorithmic}
\end{algorithm}

The algorithm maintains a fractional histogram $y$ over $\cX$, where $y_{t,j}$ denotes the histogram after the $j$th query at time $t$ has been processed. This histogram is initialized to uniform, i.e., $y_{n,0}^i = 1/N$ for all $i\in[N]$. As with static PMW, when a query is deemed hard, our algorithm performs a multiplicative weights update of $y$ with learning rate $\alpha/6$. As an extension of the static case, we also update the weights of $y$ when a new data entry arrives to reflect a data-independent prior belief 
 that data arrive from a uniform distribution.  That is, for all $t>n,i\in[N]$, $$y_{t,0}^i=\frac{t-1}{t} y_{t-1,\ell_{t-1}}^i+\frac{1}{t}\frac 1 N.$$ It is important to note that a multiplicative weights update depends only on the noisy answer to a hard query as in the static case, and the uniform update only depends on the knowledge that a new entry arrived, so this histogram can be thought of as public.
 
As in static PMW, we determine hardness using a Numeric Sparse subroutine. We specify a hardness threshold of $T=\frac{2}{3}\alpha$, and we additionally specify a function $\xi$ that varies with time and determines how much noise to add to the hardness quantities. Our most general result for $(\eps,0)$-privacy (Theorem~\ref{SPMW-changing-noise-final-thm} in Appendix \ref{app:pmw.epszero}) considers other noise functions, but for the results stated here, we let $\xi_t=\frac{\alpha^2\eps n^{1/2}}{c\log(Nn)}\cdot t^{1/2}$ for appropriate constant $c$. A query's hardness is determined by the subroutines after adding Laplace noise with parameter $4/\xi_t$. We present and analyze the required growing database modifications to Numeric Sparse and its subroutines Numeric Above Threshold and Above Threshold in Appendix~\ref{sec:backgroundalgs}; these algorithms may be of independent interest for future work in the design of private algorithms for growing databases.

We now present our main result for PMWG, Theorem~\ref{SPMW-final-thm}. We sketch its proof here and give the full proof in Appendix~\ref{app:pmw.epszero}. Whereas the accuracy results for static PMW are parametrized by the total allowed queries $k$, our noise scaling means our algorithm can accommodate more and more queries as new data continue to arrive. Our accuracy result is with respect to a query stream respecting a query budget that increases at each time $t$ by a quantity increasing exponentially with $\sqrt{t}$. This budget is parametrized by time-independent $\kappa\ge 1$, which is somewhat analogous to the total query budget $k$ in static PMW. 
This theorem tells us that PMWG can accommodate $\poly(\kappa)$ queries on the original database. Since $\kappa$ degrades accuracy logarithmically, this means we can accurately answer exponentially many queries before any new data arrive. In particular, our accuracy bounds are tight with respect to the static setting, and we maintain this accuracy through unbounded data growth, subject to a generous query budget specified by the theorem's bound on $\sum_{\tau=n}^t \ell_\tau$. 

\begin{restatable}{theorem}{epszeropmw}\label{SPMW-final-thm} 
The algorithm PMWG$(X,F,\eps,0,\alpha,n)$ is $(\eps,0)$-differentially private, and for any time-independent $\kappa\ge 1$ and $\beta>0$ it is $(\alpha,\beta)$-accurate for any query stream $F$ such that $\sum_{\tau=n}^t\ell_\tau \le \kappa \sum_{\tau=n}^t\exp(\frac{\alpha^3 \eps \sqrt{n \tau}}{C \log(Nn)})$ for all $t\ge n$ and sufficiently large constant $C$, as long as $N\ge 3, n\ge 21$ and $$\alpha\ge C\left(\frac{\log(Nn)\log(\kappa n/\beta)}{n\eps}\right)^{1/3}.$$
\end{restatable}


\begin{proof}[Proof sketch] 
 The proof hinges on showing that we do not have to answer too many hard queries, even as the composition of the database changes with new data, potentially increasing the relative entropy between the database and the public histogram. 
We first show that our new public histogram update rule bounds this relative entropy increase (Lemma~\ref{SPMW-lemma-entropy-increase}), and then show that the bound on the number of hard queries does not suffer too much relative to the static setting (Corollary~\ref{cor-bound-of-h_t}).



\begin{restatable}{lemma}{entropy}\label{SPMW-lemma-entropy-increase}
Let $x,y,\bar x,\bar y\in\Delta(\X)$ respectively be databases of size $t,t,t+1,t+1$, where $\bar x$ is obtained by adding one entry to $x$ and 
$\bar y^i = \frac{t}{t+1}y^i + \frac{1}{(t+1)N} $ for $i\in[N]$. 
 Then, 
\[\RE{\bar x}{\bar y}-\RE{x}{y} \leq  \frac{\log N}{t+1} + \frac{\log t}{t+1} + \log(\frac{t+1}{t}). \]
\end{restatable}

The corollary below comes from a straightforward modification of the proof on the bound on hard queries in static PMW using the result above. 

\begin{restatable}{corollary}{hardqbound}\label{cor-bound-of-h_t}
For a particular run of PMWG, if the Numeric Sparse subroutine returns $\alpha/3$-accurate answers for each query, then the total number of hard queries answered by any time $t\ge n$ is given by \begin{align*}
\sum_{\tau=n}^t h_\tau \leq \frac{36}{\alpha^2} \left(\log N + \sum_{\tau=n+1}^t \frac{\log(N)}{\tau} + \frac{\log (\tau-1)}{\tau} + \log(\frac{\tau}{\tau-1}) \right).
\end{align*}
\end{restatable}

With this corollary, we separately prove privacy and accuracy (Theorems~\ref{SPWM-privacy-thm} and \ref{SPWM-accuracy-thm}) in terms of the noise function $\xi$, which yield our desired result when instantiated with the $\xi$ specified by Algorithm~\ref{SPMG-alg}. As with static PMW, the only privacy leaked is by the Numeric Sparse subroutine, which is differentially private (Theorem \ref{thm.nsgpriv}). The privacy loss depends in the usual ways on the noise parameter, the query sensitivity, and the number of hard queries, although in our setting both the noise parameter and query sensitivity are changing over time. \end{proof}

After the proof Theorem \ref{SPMW-final-thm} in Appendix~\ref{app:pmw.epszero}, Theorem~\ref{SPMW-changing-noise-final-thm} gives results for a generalization of PMWG specified by Equation~(\ref{eq:SPMW-general-noise-xi-value}). This generalization leaves a free parameter in the noise function $\xi$ used by the subroutine, allowing one to trade off between accuracy and a query budget that increases more with time. 

\subsection{$(\eps,\delta)$-Differentially Private PMWG}\label{s.epsdelta}


We remark that we can tighten our accuracy bounds if we allow $(\eps,\delta)$-differential privacy in a manner analogous to the $\delta>0$ relaxation for static PMW. Using CDP composition~\cite{BS16} to bound the privacy loss at each time step in the algorithm, we get the following result.   As with its $\delta=0$ counterpart, Theorem~\ref{SPMW-delta-final-thm} is also an instantiation of a more general result (Theorem \ref{SPMW-delta-changing-noise-final-thm}) for a version of the algorithm with a free noise parameter, proven in Appendix~\ref{sec:PMWG-delta}. 

\begin{restatable}{theorem}{epsdeltapmw}\label{SPMW-delta-final-thm} 
The algorithm PMWG$(X,F,\eps,\delta,\alpha,n)$ is $(\eps,\delta)$ differentially private for any $\eps\in(0,1],\delta\in(0,e^{-1})$, and for any time-independent $\kappa\ge 1$ and $\beta\in(0,2^{-15/2})$ it is $(\alpha,\beta)$-accurate for any query stream $F$ such that $\sum_{\tau=n}^t\ell_\tau \le \kappa \sum_{\tau=n}^t\exp(\frac{\alpha^2   \epsilon\sqrt{n\tau}}{C\log^{1/2}(Nn)\log^{1/2}(1/\delta)})$ for all $t\ge n$ as long as $N\ge 3,n\ge 17$ and $$\alpha\ge C\left(\frac{\log^{1/2}(Nn)\log(\kappa n/\beta)\log^{1/2}(1/\delta)}{n\eps}\right)^{1/2}.$$
\end{restatable}

%% file: smalldbproofs.tex

\section{Analysis of Blackbox Scheduler}\label{s.bbapp}

In this appendix we present the proofs that were omitted in Section \ref{s.kevin}.   

\subsection{Fixed Accuracy as Data Accumulate}\label{s.appfixed}

\fixed*

\begin{proof}

We begin with the privacy guarantees of \FixedM{}.  When $\delta = 0$, \FixedM{} runs $\MM$ in each epoch $i$ with privacy parameter $\eps_i = \frac{\gamma^2 (i+1)}{(1+\gamma)^{i+2}}\eps$. Then by Basic Composition (\Cref{thm:basic-composition}), \FixedM{} is $(\sum_{i=0}^\infty \eps_i,0)$-differentially private, where
\begin{align*}
\sum_{i=0}^\infty \eps_i = \frac{\gamma^2}{1 + \gamma} \eps \sum_{i=0}^\infty \frac{i+1}{(1+\gamma)^{i+1}} = \eps.
\end{align*}
The sum $\sum_{i=0}^\infty \frac{i+1}{(1+\gamma)^{i+1}}$ converges to $\frac{1 + \gamma}{\gamma^2}$, so \FixedM{} is $(\eps,0)$-differentially private.

When $\delta>0$, \FixedM{} runs $\MM$ with privacy parameter $\eps_i = \frac{\gamma^{1.5}(i+1)}{(1+\gamma)^{i+1.5}}\frac{\eps}{3\sqrt{\log (1/\delta)}}$ in each epoch $i$. By \Cref{advanced-composition-eps-value}, the total privacy loss is at most
$\frac12\sum_{i=0}^\infty \epsilon_i^2+\sqrt{2\left(\sum_{i=0}^\infty \epsilon_i^2\right)\log(1/\delta)}$. Note that

\begin{align*}
\sum_{i=0}^\infty \eps_i^2 = \frac{\epsilon^2}{9\log (1/\delta)} \frac{\gamma^3}{1+\gamma}\sum_{i=0}^\infty \frac{(i+1)^2}{(1+\gamma)^{2(i+1)}} = \frac{\eps^2}{9\log (1/\delta)} \frac{\gamma^3}{1+\gamma}  \frac{(1+\gamma)^2 (\gamma^2 + 2\gamma + 2)}{\gamma^3 (\gamma+2)^3}
\le \frac{2\eps^2}{9\log (1/\delta)}
\end{align*}
where we used the fact that $\gamma \in (0,1)$. Then the total privacy loss is at most
\begin{align*}
\frac{\epsilon^2}{9 \log (1/\delta)} + \frac{2\eps}{3\sqrt{\log (1/\delta)}} \sqrt{\log (1/\delta)} \le \eps
\end{align*}
since $\eps < 1$.

To prove the accuracy of \FixedM{} we require the following lemma, which bounds the additive error introduced by answering queries that arrive mid-epoch using the slightly outdated database from the end of the previous epoch.

\additiveerr*

\begin{proof}
The linear query $x_t(f)$ can be written in the following form: $x_t(f) = \frac{1}{t}\sum_{i=1}^N t x_t^i f^i$. Then since $x_\tau(f) - x_t(f)  = \frac{1}{\tau} \sum_{i=1}^N \tau x_{\tau}^i f^i - \frac{1}{t}\sum_{i=1}^N t x_t^i f^i$, we have,
        \begin{align*}
        x_\tau(f) - x_t(f) &\le \frac{(\tau - t) + t x_t^i}{\tau} - \frac{t x_t^i}{t} \le \frac{(\tau - t) + t x_t^i}{\tau} - \frac{t x_t^i}{\tau} = 1 - \frac{t}{\tau},\\
        x_\tau(f) - x_t(f) &\ge \frac{t x_t^i}{\tau} - \frac{t x_t^i}{t} \ge t\pr{\frac{1}{\tau} - \frac{1}{t}} = \frac{t}{\tau} - 1.
        \end{align*}
        The last inequality follows because $\frac{1}{\tau} - \frac{1}{t} \le 0$ and $t x_t^i \le t$. Thus, $|x_\tau(f) - x_t(f)| \le 1-\frac{t}{\tau}$.  Since $\tau \in [t, (1+\gamma)t]$, then, 
        
        
        \begin{align*}
        1-\frac{t}{\tau} \le 1-\frac{t}{(1+\gamma)t} = 1 - \frac{1}{1+\gamma} = \frac{\gamma}{1+ \gamma}.
        \end{align*}
\end{proof}


We now continue to prove the accuracy of \FixedM{}.  Let $t_i = (1+\gamma)^i n$. Recall that epoch $i$ is defined as the time interval where $t \in \{t_i, t_i + 1, ..., t_{i+1} - 1\}$. Let $F_i$ denote the set of all queries received during epoch $i$. All queries $f \in F_i$ will be answered using $y_i$, which is computed on database $x_{t_i}$.


We want to show that $y_i(f)$ is close to $x_t(f)$ for all $f \in F_i$.
Since $y_i$ is the output of $\MM(x_{t_i},\eps_i,\alpha_i,\beta_i)$, we know that for $f\in F_i$, 
\begin{align*}
\abs{y_i(f) - x_{t_i}(f)} \le \alpha_i.
\end{align*}
By the triangle inequality and \Cref{lem:additiveerr}, for any $f\in F_i$,
\begin{align}
\abs{y_i(f)- x_t(f)} &\le \abs{y_i(f) - x_{t_i}(f)} + \abs{x_{t_i}(f)- x_t(f)} \nonumber\\
&= \alpha_i + \frac{\gamma}{1+\gamma}\label{eq:alphagamma}.
\end{align}


When $\delta=0$, we have
\begin{align*}
\alpha_i = g\pr{\frac{\log \frac{1}{\beta_i}}{\eps_i n_i}}^p  =  g\pr{\frac{(i+1) \log \frac{1+\beta}{\beta}}{\frac{\gamma^2(i+1) }{(1+\gamma)^{i+2} } \eps (1+\gamma)^i n}}^p = g \pr{\frac{(1+\gamma)^2 }{\gamma^2 \eps n}\log \frac{1+\beta}{\beta}}^p
\end{align*}
Let $Z = \frac{\log \frac{1+\beta}{\beta}}{\eps n}$. Note that since $\gamma < 1$, we have $(1+\gamma) \in (1,2)$. Then \Cref{eq:alphagamma} becomes
\begin{align*}
\cref{eq:alphagamma} \le \alpha_i + \gamma \le g Z^p\pr{\frac{1}{2}\gamma}^{-2p} + \gamma.
\end{align*}
Since we set $\gamma = g^{\frac{1}{2p+1}}Z^{\frac{p}{2p+1}}$, we have:
\begin{align*}
g Z^p\pr{\frac{1}{2}\gamma}^{-2p} + \gamma &= g Z^p \pr{\frac{1}{2} g^{\frac{1}{2p+1}}Z^{\frac{p}{2p+1}}}^{-2p} + g^{\frac{1}{2p+1}}Z^{\frac{p}{2p+1}}\\
&= C_1 g^{1 - \frac{2p}{2p+1}} Z^{p - \frac{p^2}{2p+1}} + g^{\frac{1}{2p+1}}Z^{\frac{p}{2p+1}}\\
&= C_2 g^{\frac{1}{2p+1}}Z^{\frac{p}{2p+1}}
\end{align*}
where $C_1$ and $C_2$ are positive absolute constants.

When $\delta>0$, \FixedM{} uses a different setting of $\gamma$ and $\eps_i$. Let $Z = \frac{3\sqrt{\log (1/\delta)}\log \frac{1+\beta}{\beta}}{\eps n}$. In this case, we have
\begin{align*}
\alpha_i = g\pr{\frac{\log \frac{1}{\beta_i}}{\eps_i n_i}}^p = g\pr{\frac{(i+1)\log \frac{1+\beta}{\beta}}{\frac{\gamma^{1.5}(i+1)}{{(1+\gamma)^{i+1.5}}}\frac{\eps}{3\sqrt{\log (1/\delta)}}  (1+\gamma)^i n}}^p =g Z^p\pr{\frac{1+\gamma}{ \gamma}}^{1.5p} \le g Z^p\pr{ \frac{1}{2}\gamma }^{1.5p}
\end{align*}
Since we set $\gamma = g^{\frac{1}{1.5p+1}}Z^{\frac{p}{1.5p+1}}$, we have:
\begin{align*}
\cref{eq:alphagamma} &\le g Z^p\pr{\frac{1}{2}g^{\frac{1}{1.5p+1}}Z^{\frac{p}{1.5p+1}}}^{-1.5p} + g^{\frac{1}{1.5p+1}}Z^{\frac{p}{1.5p+1}}\\
&\le C_1 g^{1 - \frac{1.5p}{1.5p+1}} Z^{p - {\frac{1.5p^2}{1.5p+1}}} + g^{\frac{1}{1.5p+1}}Z^{\frac{p}{1.5p+1}}\\
&\le C_2 g^{\frac{1}{1.5p+1}}Z^{\frac{p}{1.5p+1}}
\end{align*}
where $C_1$ and $C_2$ are positive absolute constants.

The final accuracy bound for any $\delta \geq 0$ follows by substitution and by noting that $\log \frac{1+\beta}{\beta} = O(\log \frac{1}{\beta})$ since $\beta \in (0,1)$. Each of the $\alpha_i$ bounds holds with probability $1-\beta_i$, so by a union bound, all will hold simultaneously with probability $1-\sum_{i=0}^\infty \beta_i$, where,
\begin{align*}
\sum_{i=0}^\infty \beta_i = \frac{\beta}{2 n_i^2} \le \sum_{t=n}^\infty \frac{\beta}{2t^2} \le \frac{\pi^2}{12}\beta \le \beta.
\end{align*}
Then with probability at least $1-\beta$, all queries are answered with accuracy $C g^{\frac{1}{2p+1}}Z^{\frac{p}{2p+1}}$ for $\delta = 0$ and
$C g^{\frac{1}{1.5p+1}}Z^{\frac{p}{1.5p+1}}$ for $\delta > 0$ for some positive absolute constant $C$.
\end{proof}

\subsection{Improving Accuracy as Data Accumulate}\label{s.appimprove}

%

%

\improve*

\begin{proof}

We start with the privacy guarantee. \ImprovingM{} runs $\MM$ at each time $t$ with privacy parameters $\eps_t = \frac{\eps}{t^{\frac{1}{2} + c}}$. By \Cref{advanced-composition-eps-value}, the total privacy loss is at most
$\frac{1}{2}\sum_{t=n}^\infty \epsilon_t^2+\sqrt{2\left(\sum_{t=n}^\infty \epsilon_t^2\right)\log(1/\delta)}$.

Note:
\begin{align*}
\sum_{t=n}^\infty \eps_t^2 = \frac{c\eps^2}{9\log (1/\delta)} \sum_{t=n}^\infty \frac{1}{t^{1 + 2c}} \le \frac{c\epsilon^2}{9\log (1/\delta)} \frac{1}{c n^{2c}} \le \frac{\epsilon^2}{9\log (1/\delta)}.
\end{align*}
since $n \ge 1$ and $c > 0$.
Then the privacy loss is at most:
\begin{align}
\frac{1}{18\log (1/\delta)}\epsilon^2 + \frac{2\epsilon}{3\sqrt{\log (1/\delta)}} \sqrt{\log (1/\delta)} \le \eps
\end{align} 
since $\eps < 1$. 


We next prove the accuracy of \ImprovingM{}. For each time $t$, $\MM(x_t, \eps_t, \alpha_t, \beta_t)$ is $\left(g\pr{\frac{1}{\eps_t t}}^p \log^{p''} t \log^{p'}\frac{1}{\beta_t}, \beta_t \right)$-accurate. We simply plug in $\eps_t$ and $\beta_t$ to get our accuracy bound at time $t$:
\begin{align*}
\alpha_t &=g \pr{\frac{1}{\eps_t t}}^p \log^{p''} t \log^{p'} \frac{1}{\beta_t} \\
&=g \pr{\frac{3\sqrt{\log (1/\delta)}t^{\frac{1}{2} + c}}{\sqrt{c}\eps t}}^p \log^{p''} t \log^{p'} \frac{2t^2}{\beta}\\
&\le C_1 g\pr{\frac{\sqrt{\log (1/\delta)}}{\sqrt{c}\eps t^{\frac{1}{2} - c}}}^p \log^{p'' + p'} t \log^{p'} \frac{1}{\beta}\\
&\le C_2 g\pr{\frac{\sqrt{\log (1/\delta)}}{\sqrt{c}\eps t^{\frac{1}{2} - 2c}}}^p \log^{p'} \frac{1}{\beta},
\end{align*}
for positive constants $C_1$ and $C_2$. The last line holds since $\log^{p'' + p'} t = o(t^{c})$ for any positive constants $c,p'$, and $p''$.

The accuracy of $\MM(x_t,\eps_t,\beta_t)$ at time $t$ holds with probability $1-\beta_t$. By a union bound, all accuracy guarantees will be satisfied with probability $1-\sum_{t=n}^{\infty} \beta_t$, where,
\begin{align*}
\sum_{t=n}^{\infty} \beta_t = \sum_{t=n}^\infty \frac{\beta}{2t^2} \le \frac{\pi^2}{12}\beta \le \beta.
\end{align*}

\end{proof}

%% file: pmwproofs.tex

\section{Analysis of PMW for Growing Databases}\label{app:pmw}

We first present the modification of PMW for growing databases described in Section~\ref{s.tao} and presented formally in Algorithm~\ref{SPMG-alg}. We separately prove privacy and accuracy in terms of the internal noise function $\xi$, which depends on the parameters of the algorithm. We then instantiate these theorems with our particular choice of \(\xi\) to prove Theorem~\ref{SPMW-final-thm} from the body, which gives our accuracy bound for the $(\eps,0)$-differentially private version of the algorithm. In Equation \eqref{eq:SPMW-general-noise-xi-value} we describe how to generalize our algorithm by adding a parameter for the noise function, which allows us to trade accuracy for a larger query budget. Finally, we use CDP to give our $(\eps,\delta)$-results, yielding Theorem~\ref{SPMW-delta-final-thm}.

\subsection{\((\epsilon,0)\)-Differentially Private PMWG}\label{app:pmw.epszero}

$(\eps,0)$-privacy of PMWG follows from our analysis of the NSG subroutine (Appendix~\ref{sec:nsg}) as well as Lemma~\ref{SPMW-lemma-entropy-increase} bounding the entropy increase due to new data and Corollary~\ref{cor-bound-of-h_t} bounding the number of hard queries received by any given time. We first prove this lemma and then use it with its corollary to prove privacy of PMWG in Theorem~\ref{SPWM-privacy-thm}.

\entropy*

\begin{proof}
We partition indices $i\in[N]$ into two sets $L,H$, where $i\in L$ whenever $y^i \leq \frac{1}{tN}$. For each \(S\subseteq[N]\), use \(x_{S}\) to denote \(\sum_{i\in S} x^i\). Then, by $\bar y^i \geq \frac{1}{(t+1)N}$ for all \(i\),
\begin{align*}
\sum_{i\in L}\left( \bar x^i\log(1/\bar y^i)-x^i\log(1/y^i) \right) &\leq \sum_{i\in L} \left( \bar x^i\log(N(t+1))-x^i\log(1/y^i) \right) \\
&\leq \sum_{i\in L} \left( \bar x^i\log(N(t+1))-x^i\log(tN) \right) \\
&= \sum_{i\in L} \left( \bar x^i\log(Nt)+\bar x^i\log(\frac{t+1}{t})-x^i\log(tN) \right) \\
&= \sum_{i\in L} (\bar x^i-x^{i})\log(Nt)+\bar x_L\log(\frac{t+1}{t}) \\
&\leq \sum_{i\in L} \max\{(\bar x^i-x^{i}),0\}\log(Nt)+\bar x_L\log(\frac{t+1}{t}) 
\end{align*}

The last inequality is by ignoring the term $i\in L$ with $\bar x^i < x^i$.
Next, we use $\bar y^i \geq \frac{t}{t+1}y^i$ to get

\begin{align*}
\sum_{i\in H}\left( \bar x^i\log(1/\bar y^i)-x^i\log(1/y^i) \right) &\leq \sum_{i\in H} \left( \bar x^i\log(\frac{t+1}{t}(1/y^i))-x_i\log(1/y^i) \right) \\
&= \sum_{i\in H} \left( \bar x^i\log(1/y^{i})-x_i\log(1/y^{i}) \right)+\bar x^{i}\log(\frac{t+1}{t})) \\
&= \sum_{i\in H} \left[(\bar x^i-x^{i})\log(1/y^{i})\right]+\bar x_H\log(\frac{t+1}{t}) \\
&\leq\ \sum_{i\in H} \max\{(\bar x^i-x^{i}),0\}\log(1/y^{i})+\bar x_H\log(\frac{t+1}{t}) \\
&\leq\ \sum_{i\in H} \max\{(\bar x^i-x^{i}),0\}\log(Nt)+\bar x_H\log(\frac{t+1}{t})
\end{align*}
The second inequality is by ignoring the term $i\in H$ with $\bar x^i
< x_i$.
Combining two bounds on \(L,H\) gives\begin{align*}
\sum_{i\in \X}\left( \bar x^i\log(1/\bar y^i)-x_i\log(1/y^i) \right) 
&\leq\sum_{i\in \X} \max\{(\bar x^i-x^{i}),0\}\log(Nt)+\log(\frac{t+1}{t})
\end{align*}
Since there is at most one index $i\in [N]$ such that $\bar x^i-x^i\geq 0$ (the index of newly added data entry), and for that term we have $\bar x^i-x^i = \frac{1}{t+1}+\frac{t}{t+1}x^i-x^i \leq \frac{1}{t+1}$, we have
\begin{displaymath}
\sum_{i\in \X}\left( \bar x^i\log(1/\bar y^i)-x_i\log(1/y^i) \right) \leq  \frac{1}{t+1}\log(Nt)+\log(\frac{t+1}{t})
\end{displaymath}

\end{proof}

The following corollary is an immediate extension of the analysis for static PMW using the above lemma:

\hardqbound*

\begin{theorem}[PMWG Privacy] \label{SPWM-privacy-thm}
PMWG\((X,F,\epsilon,0,\alpha,n)\) is \((\epsilon,0) \)-DP for $\xi$ as defined by the algorithm and \begin{align}
\eps = \left(1+ \frac{81}{2\alpha^2}\log N \right)\xi_n\Delta_n +\frac{81}{2\alpha^2}\sum_{t=n+1}^\infty \left( \frac{\log(N)}{t} + \frac{\log (t-1)}{t} + \log(\frac{t}{t-1}) \right) \xi_t\Delta_ t \label{SPMW-eps-value}
\end{align}
\end{theorem}
\begin{proof} 
The uniform and MW updates to the public histogram do not use any information about \(X\), and therefore these updates do not leak any privacy. The only privacy leaked is by NSG, which by Theorem \ref{SNS-privacy-theorem} is \(\epsilon\)-DP for
\begin{displaymath}
\epsilon = \xi_n\Delta_n +\frac{9}{8}\sum_{t=n}^\infty h_t\xi_t\Delta_t
\end{displaymath}
For convenience, let \(b_n=\log\ N\). Because the algorithm enforces the hard query bound from Corollary~\ref{cor-bound-of-h_t}, and because \(\xi_t\Delta_t\) is non-increasing,  
we may upper bound \(\sum_{t=n}^\infty h_t\xi_t\Delta_t\) by setting \(h_t=\frac{36}{\alpha^2}b_t\). Hence, \begin{align*}
\epsilon &\leq\left(1+ \frac{81}{2\alpha^2}\log N \right)\xi_n\Delta_n +\frac{81}{2\alpha^2}\sum_{t=n+1}^\infty \left( \frac{\log(N)}{t} + \frac{\log (t-1)}{t} + \log(\frac{t}{t-1}) \right) \xi_t\Delta_t 
\end{align*}

\end{proof}

Next we show the accuracy of PMWG. 
\begin{theorem} [PMWG Accurary]
\label{SPWM-accuracy-thm}
Let \(k:\{n,n+1,\ldots\}\rightarrow \R\). On query stream $F$ such that $\sum_{\tau=n}^t \ell_{\tau}\le \sum_{\tau=n}^t k_{\tau}$, PMWG\((X,F,\epsilon,0,\alpha,n)\) outputs answers such that \(|f_{t,j}(D_t)-a_{t,j}|\leq \alpha\) for every query $f_{t,j}$ except with probability
\begin{align}
\beta 
&\leq\exp(-\frac{\alpha \xi_n}{24}) +3 \sum_{t\geq n} k_t\exp(-\frac{\alpha \xi_t}{24}) \label{SPMW-beta-value}
\end{align}
\end{theorem}
\begin{proof}
By the exact same proof  in \cite{DR14}: PMWG's \(\alpha\)-accuracy follows if NSG returns answers that are \(\alpha/3\)-accurate.  Hence, we can take the  \(\beta \) from the NSG accuracy analysis in Theorem \ref{SNS-accurary-theorem}.\end{proof}

These privacy and accuracy results are in terms of noise scaling function $\xi$. The $\xi$ enforced by the algorithm is chosen to get good privacy and accuracy bounds simultaneously (Theorem~\ref{SPMW-final-thm}). Specifically, we need to ensure that the $\eps$ and $\beta$ bounds in \eqref{SPMW-eps-value} and \eqref{SPMW-beta-value}  converge. 
The dominating term in the hard query bound \(b_t\) in the algorithm is \(\log (t-1)/t\), so to guarantee convergence of our $\eps$ bound, we need
\begin{displaymath}
\sum_{t\geq n+1} \frac{\log(t-1)\xi_t\Delta_t}{t} =\sum_{t\geq n+1} \frac{\log(t-1)\xi_t}{t^2} \leq \int_{t=n}^\infty\frac{\log(t-1)\xi_t}{t^2} dt, 
\end{displaymath}
so we may pick \(\xi=O(t^{1-c})\) for any \(c>0\). To guarantee convergence of our $\beta$ bound and allow $\alpha$ to degrade logarithmically in the number of queries $\ell_t$, we want $
  \sum_{t\geq n} k_t\exp(-\frac{\alpha \xi_t}{24})
$
 to converge exponentially quickly in \(\alpha\). For example, we may pick query budget \(k_t=O(\exp(\frac{\alpha \xi_t}{48}))\) and \(\xi_t=\Omega(t^c)\) for any $c>0$. 
To get these conditions simultaneously, Algorithm~\ref{SPMG-alg} picks \(\xi_t=ct^{1/2}\) for some time-independent quantity $c>0$, and we guarantee accuracy for query streams respective query budget \(k_t=\kappa \exp(\frac{\alpha c t^{1/2}}{48})\) for some constant \(c,\kappa\). Other reasonable choices of $\xi_t,k_t$ are possible; we discuss these after proving our main result for Algorithm~\ref{SPMG-alg}.

We prove the result below without suppressing constants in the query budget and the $\alpha$ bound, i.e., we prove that $\sum_{\tau=n}^t \ell_{\tau} \le \kappa \sum_{\tau=n}^t \exp(\frac{\alpha^3\eps\sqrt{n\tau}}{8262\log(Nn)})$ and $\alpha\ge (\frac{8262\log(Nn)\log(192\kappa n/\beta)}{n\eps})^{1/3}$ suffice for accuracy. Note that with this choice of \(\alpha\), the query budget grows exponentially with $\sqrt t$ since 
$\kappa\sum_{\tau=n}^t \exp\left(\frac{\alpha^3 \eps \sqrt{n \tau}}{8262\log(Nn)}\right)\geq \kappa\sum_{\tau=n}^t \left(\frac{192\kappa n}{\beta}\right)^{\frac{17}{16}\sqrt{\frac{\tau}{n}}}.$

\epszeropmw*

\begin{proof} Our main result is an instantiation of the more general results in Theorems \ref{SPWM-privacy-thm} and \ref{SPWM-accuracy-thm}. In what follows, let $\xi_t=ct^{1/2}$ for some time-independent $c>0$. 
Applying Theorem \ref{SPWM-privacy-thm}, the  privacy loss
of PMWG is 
\begin{align*}
\eps' &= \left(1+ \frac{81}{2\alpha^2}\log N \right) cn^{-1/2}+\frac{81}{2\alpha^2}\sum_{t=n+1}^\infty \left( \frac{\log(N)}{t} + \frac{\log (t-1)}{t} + \log(\frac{t}{t-1}) \right) ct^{-1/2} \\
&\leq\left( \frac{81}{\alpha^2}\log N \right) cn^{-1/2}+\frac{81c}{2\alpha^2}\int_{t=n}^\infty \left( \frac{\log(N)}{t^{3/2}} + \frac{\log (t)}{t^{3/2}} + \frac{1}{(t-1)^{3/2}}  \right)  dt \\
 &= \frac{81c\log N}{\alpha^2n^{1/2}} +\frac{81c}{2\alpha^2}\left[ -\frac{2\log N}{t^{1/2}} - \frac{4+2\log t}{t^{1/2}} - \frac{2}{(t-1)^{1/2}}\right]_{t=n}^\infty \\
&=\frac{81c}{\alpha^2}\left( \frac{2\log N}{n^{1/2}} + \frac{2+\log n}{n^{1/2}} + \frac1{(n-1)^{1/2}}\right) \\
&\leq\frac{81c}{\alpha^2}\cdot \frac{2(\log N + \log n)}{n^{1/2}},
\end{align*}
using \(\log(t/(t-1))t^{-1/2}=\log(1+\frac{1}{t-1})t^{-1/2}\leq \frac{1}{t-1}(t-1)^{-1/2}\) for the first inequality and  \(n\geq 21\) for the last inequality. 
Finally, setting \(c=\frac{\alpha^2 n^{1/2}}{162\log(Nn)}\epsilon\)  gives \(\epsilon'=\epsilon\).

Applying  Theorem \ref{SPWM-accuracy-thm},  PMWG is \((\alpha,\beta')\) -accurate for
\begin{align*}
\beta' &=\exp(-\frac{c\alpha n^{1/2}}{24}) +3 \kappa \sum_{t\geq n} \exp(\frac{c\alpha t^{1/2}}{48})\exp(-\frac{c\alpha t^{1/2}}{24}) \\
&\leq\exp(-\frac{c\alpha n^{1/2}}{24})  +3\kappa \int_{t=n-1}^\infty \exp(-\frac{c\alpha t^{1/2}}{48})dt \\
&=\exp(-\frac{c\alpha n^{1/2}}{24})  +3\kappa  \left[-\dfrac{96\left(\alpha c\sqrt{t}+48\right)\exp(-\frac{\alpha c\sqrt{t}}{48})}{\alpha^2c^2}\right]_{t=n-1}^\infty \\
&=\exp(-\frac{c\alpha n^{1/2}}{24})+288\kappa \dfrac{\left(\alpha c\sqrt{n-1}+48\right)}{\alpha^2c^2} \exp(-\frac{\alpha c\sqrt{n-1}}{48})
\end{align*}
To get \(\beta'\leq \beta\), we can require \(\exp(-\frac{c\alpha n^{1/2}}{24})\leq \beta/2\) and \(288\kappa \dfrac{\left(\alpha c\sqrt{n-1}+48\right)}{\alpha^2c^2} \exp(-\frac{\alpha c\sqrt{n-1}}{48}) \leq \beta/2\). The first is equivalent to
\begin{align}
\frac{c\alpha n^{1/2}}{24} \geq \log (2/\beta) \iff c\alpha\geq \frac{24\log(2/\beta)}{n^{1/2}} \iff \alpha \geq \left(\frac{3888\log(Nn)\log(2/\beta)}{n\epsilon}\right)^{1/3} \label{SPMW-alpha-value-1}
\end{align}
We assume that \(\alpha\) satisfies  \eqref{SPMW-alpha-value-1} before proceeding. Secondly, \eqref{SPMW-alpha-value-1} gives \(c\alpha\geq \frac{24\log(2/\beta)}{n^{1/2}}\), which implies \[\dfrac{\left(\alpha c\sqrt{n-1}+48\right)}{\alpha^2c^2}\leq \frac{n^{1/2}\sqrt{n-1}}{24\log(2/\beta)}+\frac{n}{12\log^2(2/\beta)}\leq \frac{n(\log(2/\beta)+2)}{24\log^2(2/\beta)}\leq\frac{n(2+\log 2)}{24\log^2(2)}<\frac{n}{3}\]
Hence, it's enough to require \(96\kappa n\exp(-\frac{\alpha c\sqrt{n-1}}{48}) \leq \beta/2\). This is equivalent to
\begin{align*}
\frac{\alpha c\sqrt{n-1}}{48} \geq \log(192\kappa n/\beta)
\end{align*} For \(n\geq 9\), \(\frac{\sqrt{n-1}}{48} \geq \frac{n^{1/2}}{51}\), so we only need 
\begin{align}
\frac{\alpha c}{51} \geq \frac{\log(192\kappa n/\beta)}{n^{1/2}} \iff \alpha\geq\left( \frac{8262\log(Nn)\log(192\kappa n/\beta)}{n\epsilon} \right)^{1/3} \label{SPMW-alpha-bound2}
\end{align}
\eqref{SPMW-alpha-bound2} is a stronger bound than \eqref{SPMW-alpha-value-1} for \(\kappa \geq 1\).

\cut{
We now state in form of \(k_n\). The requirement \eqref{SPMW-alpha-bound2} becomes 

\begin{align*}
\frac{\alpha c}{17} \geq \frac{\log(192k_nn/\beta)-\frac{\alpha c n^{1/2}}{16}}{n^{1/2}} &\iff c_0^{} \alpha c\geq\frac{\log(192k_nn/\beta)}{n^{1/2}} \\
&\iff\ \frac{c_0\alpha^3 n^{1/2}}{18\log(Nn)}\epsilon \geq\frac{\log(192k_nn/\beta)}{n^{1/2}} \\
&\iff\alpha \geq  \left( \frac{18c_0^{-1}\log(Nn)\log(192k_nn/\beta)}{n\epsilon} \right)^{1/3}  
\end{align*}
 for \(c_0=1/16+1/17 >1/ 9 \), proving \eqref{SPMW-alpha-final-value2}.
It is obvious that this bound dominates \(\left(\frac{144\log(Nn)\log(2/\beta)}{n\epsilon}\right)^{1/3}\) for any \(k_n \geq 1\).}
\end{proof}

\subsubsection{$(\eps,0)$-Differential Privacy for PMWG with Different Noise Functions}\label{sec:gennoise}
We can still achieve good privacy and accuracy for other choices of $\xi$ and query budget. In this section we consider a generalization of PMWG that include an additional parameter \(p\). The only difference between this generalization and Algorithm~\ref{SPMG-alg} is that the modified version defines the noise function as
\begin{equation}
\xi_t=\begin{cases} \frac{\alpha^2 (1-p)^2n^{1-p}}{126\log(Nn)}\epsilon t^p & \text{ if } \delta=0 \\ \frac{\alpha(1-p)n^{1-p}}{24\log^{1/2} (Nn)\log^{1/2}(1/\delta)}\epsilon t^p 
& \text{ if }\delta>0 \end{cases} \label{eq:SPMW-general-noise-xi-value}
\end{equation}
Theorem~\ref{SPMW-changing-noise-final-thm} shows that this choice yields an accuracy bound of 
$$\alpha \geq \left(\frac{6426\log(Nn)\log(144\kappa n/\beta)}{(1-p)^2n\epsilon}\right)^{1/3}.$$
This is comparable to accuracy of static PMW, $
\alpha_{\text{static}}=\Theta\left(\frac{\log N \log(k/\beta)}{ n\epsilon}\right)^{1/3}$
, and it yields query budget
\begin{equation}
\kappa \sum_{\tau=n}^t\exp\left(\frac{\alpha^3 (1-p)^2\eps n^{1-p}\tau^p}{6048 \log(Nn)}\right)\geq \kappa\sum_{\tau=n}^t \left(\frac{144\kappa n}{\beta}\right)^{\frac{17}{16}\left(\frac{\tau}{n}\right)^p}. \label{SPMW-general-noise-query-budget-no-alpha}
\end{equation}
These bounds tell us that as $p$ approaches 1, the query budget approaches exponential in $t$, but accuracy suffers proportionally to $(1-p)^{-2/3}$. 
Therefore, we suffer only a constant loss in accuracy as long as \(n\) is bounded polynomially in \(N\) and \(k/\beta\). Note that, however, our query budget in the growing setting allows a  generous additional number of queries to be asked upon each arrival of new data entry. 

Before proving our general result (Theorem~\ref{SPMW-changing-noise-final-thm}), we first present two useful lemmas. The first bounds privacy loss for $\xi$ of our more general form.
\begin{lemma}  \label{SPMW-bound-eps-different-noise}
Assuming $n\ge 5, N\ge 3$, the privacy loss associated with Algorithm~\ref{SPMG-alg} using $\xi_t$ such that $\xi_t\Delta_t = ct^{-q}$ for some constant \(c\) independent of \(t\) and \(1 \geq q >0\) is 
$
\epsilon'\leq \frac{126c\log(Nn)}{\alpha^2 q^2n^q}.
$
\end{lemma}
\begin{proof} Applying the hard query bound from Corollary~\ref{cor-bound-of-h_t}, privacy loss is bounded as follows:
\begin{align*}
\eps' &= \left(1+ \frac{81}{2\alpha^2}\log N \right)  \xi_n\Delta_n   +\frac{81}{2\alpha^2}\sum_{t=n+1}^\infty \left( \frac{\log(N)}{t} + \frac{\log (t-1)}{t} + \log(\frac{t}{t-1}) \right)  \xi_t\Delta_t \\
 &=  \left(1+ \frac{81}{2\alpha^2}\log N \right) c n^{-q}+\frac{81}{2\alpha^2}\sum_{t=n+1}^\infty \left( \frac{\log(N)}{t} + \frac{\log (t-1)}{t} + \log(\frac{t}{t-1}) \right) ct^{-q} \\
&\leq\left( \frac{81}{\alpha^2}\log N \right) c n^{-q}+\frac{81c}{2\alpha^2}\int_{t=n}^\infty \left( t^{-q-1}\log(N) + t^{-q-1}\log(t)+(t-1)^{-q-1} \right)  dt \\
 &= \frac{81c\log N}{\alpha^2} n^{-q}+\frac{81c}{2\alpha^2}\left[ -\frac{\log(N)t^{-q}+\log(t)t^{-q}+(t-1)^{-q}}{q} - \frac{t^{-q}}{q^2}
\right]_{t=n}^\infty \\
&\leq\frac{81c\log N}{\alpha^2} n^{-q}+\frac{81c}{2\alpha^2}\left(  \frac{\log(N)n^{-q}+\log(n)n^{-q}+2n^{-q}}{q} + \frac{n^{-q}}{q^2}\right) \\
&\leq\frac{81c\log N}{\alpha^2} n^{-q}+\frac{81c}{2\alpha^2}\left(  \dfrac{\log(N)+\log(n)+3}{q^2} \right) n^{-q} \\
&\leq\frac{81cn^{-q}}{2\alpha^2}\left(2\log N +   \dfrac{\log(N)+\log(n)+3}{q^2}\right) \\
&\leq\frac{81cn^{-q}}{2\alpha^2}\left( \frac{3\log (Nn)}{q^2}\right) \\
&\leq \frac{126c\log(Nn)}{\alpha^2 q^2n^q}
\end{align*}
where we use the fact that \((n-1)^{-q} \leq 2n^{-q}\) and \(\log N \leq \frac{\log N}{q^2}\)for \(n \geq 2,0\leq q\leq1\),  and that \(3 \leq 2\log n\) for \(n \geq 5\).
\end{proof}
 We also know that for \(\xi_t=\Omega(t^p)\) for some \(p>0\), \(\int_{t=n-1}^\infty \exp(-\frac{c\alpha\xi_t}{48})dt=P_{p,c\alpha}(n-1)\exp(-\frac{c\alpha\xi_{n-1}}{48})\) for some polynomial \(P_{p,c\alpha}(n)\) dependent on \(p,c\alpha\). Stating the exact bound for this integral, however, involves approximating an upper incomplete gamma function. In the following lemma, which will be used to bound the probability of failure for our algorithm with these modified parameters, we therefore require \(p\geq 1/4\), noting that we could better optimize constants with smaller $p$.
 
\begin{lemma} \label{SPMW-bound-beta-different-noise}
Let \(1\geq p\geq 1/4,n\geq 17\) and \(c,\alpha\) be constants independent of \(t\) such that \(c\alpha n^p \geq 24\log(2/\beta) \) and \(\beta < 2^{-15/2}\). Then

\begin{align*}
\int_{t=n-1}^\infty \exp\left(-\frac{c\alpha t^{p}}{48}\right)dt \leq \frac{6ne^{-\frac{c\alpha n^p}{51}}}{p}
\end{align*}\end{lemma}
\begin{proof}

We have
\begin{align*}
\int_{t=n-1}^\infty \exp(-\frac{c\alpha t^{p}}{48})dt =\dfrac{\operatorname{\Gamma}\left(\frac{1}{p},\frac{c\alpha}{48} (n-1)^p\right)}{p(\frac{c\alpha}{48})^\frac{1}{p}}  \\
\end{align*}
where \(\Gamma(s,x)=\int_x^\infty t^{s-1}e^{-t}dt\) is an upper incomplete gamma function. We now use the bound in \cite{J16}\footnote{Proposition 10 from the extended version  http://www.maths.lancs.ac.uk/jameson/gammainc.pdf} that for any \(a>1,e^x > 2^a \), we have \( \operatorname{\Gamma}\left(a,x\right)\leq 2^ax^{a-1}e^{-x} \). With \(n\geq 17\), \(\frac{c\alpha}{48} (n-1)^p\geq \frac{c \alpha n^p}{51} \geq \frac{8\log(2/\beta)}{17}\). Choosing \(\beta < 2^{-15/2}\) gives \(\exp(\frac{8\log(2/\beta)}{17} ) \geq 2^4 \geq 2^{1/p}\), so we can apply the bound \begin{align*}
\operatorname{\Gamma}\left(\frac{1}p,\frac{c\alpha}{48} (n-1)^p\right) \leq\operatorname{\Gamma}\left(\frac{1}{p},\frac{c\alpha n^p}{51} \right) \leq2^{1/p}\left (\frac{c\alpha n^p}{51}\right)^{\frac{1}{p}-1}e^{-\frac{c\alpha n^p}{51}} 
\end{align*}
Therefore,
\begin{align*}
\dfrac{\operatorname{\Gamma}\left(\frac{1}{p},\frac{c\alpha}{48} (n-1)^p\right)}{p(\frac{c\alpha}{48})^\frac{1}{p}} &\leq\frac{51(32/17)^{1/p}ne^{-\frac{c\alpha n^p}{51}}}{pc\alpha n^p} \leq \frac{17(32/17)^{1/p}ne^{-\frac{c\alpha n^p}{51}}}{8p\log(2/\beta)} \\
&\leq\frac{6ne^{-\frac{c\alpha n^p}{51}}}{p} 
\end{align*}

where we use \(p\geq 1/4\) and \(\beta < 2^{-15/2}\) to get the last inequality.
\end{proof}

Finally, we give our overall privacy and accuracy guarantees of our modified algorithm.

\begin{theorem}[Generalized-Noise PMWG \(\epsilon\)-DP 
Result] \label{SPMW-changing-noise-final-thm}
Let \(p\in [1/4,1)\). The algorithm PMWG$(X,F,\eps,0,\alpha,n,p)$ that generalizes Algorithm~\ref{SPMG-alg} as in \ref{eq:SPMW-general-noise-xi-value} is $(\eps,0)$-differentially private, and for any time-independent $\kappa\ge 1$ and $\beta\in(0, 2^{-15/2})$ it is $(\alpha,\beta)$-accurate for any query stream $F$ such that $\sum_{\tau=n}^t\ell_\tau \le \kappa \sum_{\tau=n}^t\exp(\frac{\alpha^3 (1-p)^2\eps n^{1-p}\tau^p}{6048 \log(Nn)})$ for all $t\ge n$  as long as $N\ge 3, n\ge 17$ and
\begin{align}
\alpha \geq \left(\frac{6426\log(Nn)\log(144\kappa n/\beta)}{(1-p)^2n\epsilon}\right)^{1/3} \label{SPMW-eps-changing-noise-alpha-final-value}
\end{align}
\end{theorem}

\begin{proof}
Applying Theorem \ref{SPWM-privacy-thm} and Lemma \ref{SPMW-bound-eps-different-noise}, the  privacy loss
of PMWG is 
\begin{align*}
\eps' &\leq \frac{126c\log(Nn)}{\alpha^2 q^2n^q} 
\end{align*}
where \(q=1-p> 0\).
Setting \(c=\frac{\alpha^2 (1-p)^2n^{1-p}}{126\log(Nn)}\epsilon\)  gives \(\epsilon'=\epsilon\).

Applying  Theorem \ref{SPWM-accuracy-thm},  PMWG is \((\alpha,\beta')\) -accurate for
\begin{align*}
\beta' &\leq\exp(-\frac{c\alpha n^{p}}{24})  +3\kappa \int_{t=n-1}^\infty \exp(-\frac{c\alpha t^{p}}{48})dt \\
&
\end{align*}
Again, we require the first term to be at most \(\beta/2\):
\begin{align}
\frac{c\alpha n^{p}}{24} \geq \log (2/\beta) \iff \frac{(1-p)^2\alpha^3n}{126\log(Nn)}\epsilon\geq 24\log(2/\beta) \iff \alpha \geq \left(\frac{3024\log(Nn)\log(2/\beta)}{(1-p)^2n\epsilon}\right)^{1/3} \label{SPMW-changingell-alpha-value-1}
\end{align}We assume that \(\alpha\) satisfies this requirement before proceeding. Now we  apply  Lemma \ref{SPMW-bound-beta-different-noise}:

\begin{align*}
\int_{t=n-1}^\infty \exp(-\frac{c\alpha t^{p}}{48})dt &\leq\frac{6ne^{-\frac{c\alpha n^p}{51}}}{p} \leq24ne^{-\frac{c\alpha n^p}{51}} 
\end{align*}
So it's enough to require \(3\kappa \cdot24ne^{-\frac{c\alpha n^p}{51}}\leq \beta/2\). This is equivalent to
\begin{align}
\frac{c\alpha n^p}{51} \geq \log \left(\frac{144\kappa n}{\beta}\right)  &\iff \frac{(1-p)^2\alpha^3n}{126\log(Nn)}\epsilon\geq 51\log\left(\frac{144\kappa n}{\beta}\right) \label{SPMW-changing-noise-alpha-bound-2} \\
&\iff \alpha \geq \left(\frac{6426\log(Nn)\log(144\kappa n/\beta)}{(1-p)^2n\epsilon}\right)^{1/3}
\end{align} 
 Clearly the second requirement is stronger than the first for \(\kappa \geq 1\).
\cut{
We now state in form of \(k_n\). The requirement \eqref{SPMW-changing-noise-alpha-bound-2} becomes 

\begin{align*}
\frac{c\alpha n^p}{17} \geq \log \left(\frac{144k_nn}{\beta}\right)-\frac{c\alpha n^p}{16}&\iff c_0 \alpha cn^p\geq \log(144k_nn/\beta) \\
&\iff\ \frac{c_0(1-p)^2\alpha^3n}{14\log(Nn)}\epsilon \geq \log(144k_nn/\beta) \\
&\iff\alpha \geq  \left( \frac{14c_0^{-1}\log(Nn)\log(144k_nn/\beta)}{(1-p)^2n\epsilon} \right)^{1/3}  
\end{align*}
 for \(c_0=1/16+1/17 >1/ 9 \), proving \eqref{SPMW-changing-noise-alpha-final-value2}. It is obvious that this bound dominates the earlier bound on \(\alpha\) for any \(k_n \geq 1\)
}
\end{proof}

\subsection{\((\epsilon,\delta)\)-Differentially Private PMWG} \label{sec:PMWG-delta}

With Theorem \ref{advanced-composition-eps-value}, the total privacy loss from composition of differentially private subroutines is the \textit{sum of squared} privacy losses, rather than the sum.
For our algorithm with $\delta>0$, it is straight forward to compute the sum of squares of
all points of privacy loss as 

\begin{align}\label{eq:ssq-delta}
\eps' :=\left(1+ \frac{585}{16\alpha^2}\log N \right)  (\xi_n\Delta_n  )^2 +\frac{585}{16\alpha^2}\sum_{t=n+1}^\infty \left( \frac{\log(N)}{t} + \frac{\log (t-1)}{t} + \log(\frac{t}{t-1}) \right)  (\xi_t\Delta_t)^2
\end{align}
We now state how to upper bound this sum, which is in a similar form to Lemma \ref{SPMW-bound-eps-different-noise}.
\begin{lemma}  \label{SPMW-bound-eps-delta-different-noise}
Let \(\xi_t\Delta_t = ct^{-q}\) for some constant \(c\) independent of \(t\) and \(1 \geq q >0\). Then, for \(n\geq3\), \(N\geq3\), and $\eps'$ as in Equation~\ref{eq:ssq-delta}, we have
$\sqrt{\eps'} \leq \frac{12cn^{-q} \log^{1/2} (Nn)}{\alpha q}
$.
\end{lemma}
\begin{proof}
By a similar calculation as in Lemma \ref{SPMW-bound-eps-different-noise} but with \(c^2n^{-2q}\) in place of \(cn^{-q}\),
\begin{align*}
\eps' &\leq\frac{585c^2\log N}{8\alpha^2} n^{-2q}+\frac{585c^2}{16\alpha^2}\left(  \frac{\log(N)n^{-2q}+\log(n)n^{-2q}+2n^{-2q}}{2q} + \frac{n^{-2q}}{4q^2}\right) \\
&\leq\frac{585^2\log N}{8\alpha^2} n^{-2q}+\frac{585c^2}{16\alpha^2}\left(  \dfrac{2\log(N)+2\log(n)+5}{4q^2} \right) n^{-2q} \\
&\leq\frac{c^2n^{-2q}}{\alpha^2}\left(\frac{585}{8}\cdot \frac{\log N}{q^2}+\frac{585}{64}\cdot   \dfrac{2\log(N)+10\log(n)}{q^2}\right) \\
&=\frac{c^2n^{-2q}}{\alpha^2}\left( \frac{2925\log (Nn)}{32q^2}\right)
\end{align*}
where we use the bounds \(5\leq 8\log(n)\) and \(\log N \leq \frac{\log N}{q^2}\). The result now follows.
\end{proof}

With the bound by Lemma \ref{SPMW-bound-eps-delta-different-noise} and  Theorem \ref{advanced-composition-eps-value}, we  achieve an \((\eps,\delta)\)-differentially private version of PMWG. Note that Theorem~\ref{SPMW-delta-final-thm} instantiates this result with $p=1/2$.

\begin{theorem}[Generalized-Noise PMWG \((\epsilon,\delta)\)-DP  Result] \label{SPMW-delta-changing-noise-final-thm}
\cut{Let \(\epsilon\in(0,1],\beta\in(0,2^{-15/2} ),\delta\in(0,e^{-1}),p\in [1/4,1)\),  \(\xi_t=ct^{p}\) for \(c= \frac{\alpha(1-p)n^{1-p}}{24\log^{1/2} (Nn)\log^{1/2}(1/\delta)}\epsilon \), and stream of queries \(F\) that respects the query budget \(\sum_{t_0=n}^t k_{t_0}\), where \(k_t=\kappa \exp(\frac{\alpha ct^{p}} {48})\). Let \(N\geq 3\) and \(n\geq 17\). Then PMWG\((D,F,\eps,\alpha)\), is \((\epsilon,\delta)\)-DP and \((3\alpha,\beta)\)-accurate for

\begin{align}
\alpha\geq\max \left\{ \left(\frac{64\log^{1/2}(Nn)\log(2/\beta)\log^{1/2}(1/\delta)}{(1-p)n\epsilon}\right)^{1/2} , \left(\frac{136\log^{1/2}(Nn)\log(144\kappa n/\beta)\log^{1/2}(1/\delta)}{(1-p)n\epsilon}\right)^{1/2} \right\} \label{SPMW-delta-changing-noise-alpha-final-value}
\end{align}
Alternatively, in terms of \(k_n\) instead of \(\kappa \), we have
that for all \(k_n \geq 1,\)
\begin{align}
\alpha\geq  \left(\frac{72\log^{1/2}(Nn)\log(144k_nn/\beta)\log^{1/2}(1/\delta)}{(1-p)n\epsilon}\right)^{1/2} \label{SPMW-delta-changing-noise-alpha-final-value2}
\end{align}
}
Let \(p\in [1/4,1),\delta\in(0,e^{-1})\). The algorithm PMWG$(X,F,\eps,\delta,\alpha,n,p)$ is $(\eps,\delta)$-differentially private, and for any time-independent $\kappa\ge 1$ and $\beta\in(0, 2^{-15/2})$ it is $(\alpha,\beta)$-accurate for any query stream $F$ such that $\sum_{\tau=n}^t\ell_\tau \le \kappa \sum_{\tau=n}^t\exp(\frac{\alpha^2 (1-p)\eps n^{1-p}\tau^p}{1152 \log^{1/2} (Nn)\log^{1/2}(1/\delta)})$ for all $t\ge n$  as long as $N\ge 3, n\ge 17$ and
\begin{align}
\cut{\alpha\geq\max \left\{ \left(\frac{112\log(Nn)\log(2/\beta)}{(1-p)^2n\epsilon}\right)^{1/3} , \left(\frac{238\log(Nn)\log(144\kappa n/\beta)}{(1-p)^2n\epsilon}\right)^{1/3} \right\}}
\alpha \geq \left(\frac{1224\log^{1/2}(Nn)\log(144\kappa n/\beta)\log^{1/2}(1/\delta)}{(1-p)n\epsilon}\right)^{1/2} \label{SPMW-delta-changing-noise-alpha-final-value}
\end{align}
\end{theorem}
\cut{Note that with this choice of \(\alpha\), we have \(\alpha c \geq \frac{9\log(144k_nn/\beta)}{n^{p}}\),  so \[k_t= k_n\exp(\frac{\alpha ct^{p}} {16}-\frac{\alpha cn^{p}} {16})\geq k_n\exp\left(\frac{9\log(144k_nn/\beta) \left((t/n)^{p}-1\right)}{16}\right) = k_n\left(\frac{144k_nn}{\beta}\right)^{\frac{9}{16}\left(\left(\frac{t}{n}\right)^p-1\right)} \]
}

\begin{proof}
\cut{We analyze the composition of privacy loss of SNAT, SNS, and SPMW to obtain the following results (analogous to Theorem \ref{SNAT-privacy-theorem}, \ref{SNS-privacy-theorem}, and \ref{SPWM-privacy-thm}):
\begin{itemize}
\item 
 SNAT is \((\frac{1}{2}(\xi_{t_0}\Delta_{t_0})^2 +\ \frac{1}{128} (\xi_{t'}\Delta_{t'})^2)\)-zCDP
\item 
SNS is \((\frac{1}{2}(\xi_n\Delta_n)^2+\frac{65}{128}\sum_{t=n}^\infty h_t(\xi_t\Delta_t)^2\)-zCDP
\item 
SPMW is \(\left(\left(\frac{1}{2}+\frac{9}{4\alpha^2}\log N\right)(\xi_n\Delta_n)^2 + \frac{65}{32\alpha^2}\sum_{t=n}^\infty b_t(\xi_t\Delta_t)^2\right)\)-zCDP
\end{itemize} 
}
Applying Lemma \ref{SPMW-bound-eps-delta-different-noise} and  Theorem \ref{advanced-composition-eps-value},  PMWG is \((\eps,\delta)\)-DP for
\begin{align*}
\eps' &\leq\frac{24cn^{-q} \log^{1/2} (Nn)\log^{1/2}(1/\delta)}{\alpha q}  
\end{align*}
where \(q=1-p> 0\).
Setting \(c=\frac{\alpha(1-p)n^{1-p}}{24\log^{1/2} (Nn)\log^{1/2}(1/\delta)}\epsilon\)  gives \(\epsilon'=\epsilon\). The rest of the proof now follows the proof of Theorem \ref{SPMW-changing-noise-final-thm}, except that now \(\alpha c n^p = \frac{\alpha^2 (1-p) n \epsilon}{24\log^{1/2}(Nn)\log^{1/2}(1/\delta)}\) instead of \(\alpha c n^p = \frac{\alpha^3 (1-p)^2 n \epsilon}{126 \log(Nn)}\).
\end{proof}

Finally, we can make an observation analogous to that made at the beginning of Subsection~\ref{sec:gennoise}  for $(\eps,\delta)$-differential privacy. Note that for \(\alpha\) as in Theorem~\ref{SPMW-delta-changing-noise-final-thm}, we have query budget \begin{equation}
\kappa \sum_{\tau=n}^t\exp(\frac{\alpha^2 (1-p)\eps n^{1-p}\tau^p}{1152 \log^{1/2} (Nn)\log^{1/2}(1/\delta)})\geq \kappa\sum_{\tau=n}^t \left(\frac{144\kappa n}{\beta}\right)^{\frac{17}{16}\left(\frac{\tau}{n}\right)^p} \label{SPMW-delta-query-budget}
\end{equation}
As we let \(p\) approach 1, we increase the query budget (lower bounded by \eqref{SPMW-general-noise-query-budget-no-alpha} as well) and  suffer accuracy loss proportional to \((1-p)^{-{1/2}}\) 
 compared to $\Theta\left(\frac{\log^{1/2} N \log(k/\beta)\log(1/\delta)}{\epsilon n}\right)^{1/2}$ accuracy  of static $(\eps,\delta)$-private PMW, as long as \(n\) is bounded polynomially in \(N\) and \(k/\beta\).

%% file: app.tex
\section{Sparse Vector Algorithms for Growing Databases}\label{sec:backgroundalgs}

In this section, we consider three primitive algorithms that apply the sparse vector technique in the static setting: above threshold, numeric above threshold, and numeric sparse. We modify these algorithms  for the dynamic setting of growing databases, and we call our modifications ATG, NATG, and NSG, respectively. The main difference in analyzing privacy and accuracy for dynamic algorithms is that the results now depend on a changing database size. In this section, the database has starting size $t_0$ and it grows by one entry each time step, starting at time $t_0$. 

\subsection{Above Threshold for Growing Databases}
We first consider the simplest of these algorithms: ATG. This algorithm simply compares a noisy answer to a query from stream $F=\set{\set{f_{t,i}}_{i=1}^{\ell_t}}_{t\ge t_0}$ applied to the current database $D_t$ to a noisy threshold $\hat T_t \approx T$ and aborts the first time a noisy answer exceeds the noisy threshold. The noise added is determined by a noise function $\xi:\set{t_0,t_0+1,\dots}\to \R$.

\begin{algorithm} \label{ATG-alg-code}
\caption{\textsc{ATG}($D,F,T,\xi$)} 
\begin{algorithmic}
\For {each incoming query \(f_{t,i}\)}
\If {$i=1$}
        \State Let \(\hat{T}_t\la T+\Lap(\frac{2}{\xi_t})\) \Comment Noisy threshold for queries at time \(t\)
        \EndIf
\State Let \(\nu_{t,i}\la \Lap(\frac{4}{\xi_t})\) \Comment{Noise for query $f_{t,i}$}

\If {$f_{t,i}(D_t)+\nu_{t,i}\geq \hat{T}_t$}
    \State Output $a_{t,i} \la \top$ and halt \Comment{Abort on first query above threshold}
\Else
    \State Output \(a_{t,i}\la \bot\)
\EndIf

\EndFor

\end{algorithmic}
\end{algorithm}

\begin{theorem}[Privacy of ATG]
\label{SAT-privacy-theorem}
Let \(\xi,\Delta:\{t_0,t_0+1,\ldots\}\rightarrow \R^+\) be such that both \(\Delta,\xi \cdot  \Delta\) are non-increasing and \(\xi\) is non-decreasing. Let $F=\{\{f_{t,i}\}_{i=1}^{\ell_t}\}_{t\ge t_0}$ be a stream of queries with sensitivity \(\Delta_{f_{t,i}}\leq \Delta_t\) for all \((t,i)\). Then for any  \(a\) in the range of ATG and any neighboring database streams \(D,D'\) of starting size $t_0$,

\begin{displaymath}
\Prob{}{\text{ATG}(D,F,T,\xi)= a}\leq \exp(\xi_{t_0}\Delta_{t_0})\Prob{}{\text{ATG}(D',F,T,\xi)=a}
\end{displaymath}

\end{theorem}
\begin{proof}
We follow the proof as in \cite{DR14}, but modify it slightly. Let \(A,A'\) represent the random variables of output of ATG running on \(D,D'\), respectively. Suppose ATG halts at the last query \(f_{t',i'}\). Let \(a\) denote this output, i.e. \(a_{t,i}=\bot\) for all \((t,i)<(t',i') \)  and \(a_{t',i'}=\top\). Define
\begin{displaymath}
H(D)=\max_{(t,i)<(t',i')} \xi_t\cdot(f_{t,i}(D_t) + \nu_{t,i}-T)
\end{displaymath}
Fix \(\nu_{t,i}\) for all \((t,i)<(t',i')\) , so that \(H(D)\) is a deterministic quantity. Then,
\begin{align*}
\Prob{\eta,\nu_{t',i'}}{A=a}&=\Prob{\eta,\nu_{t',i'}}{f_{t,i}(D_t)+\nu_{t,i}< \hat{T}_t,\ \forall(t,i)<(t',i') \text{ and } f_{t',i'}(D_{t'})+\nu_{t',i'}\geq \hat{T}_{t'}} \\
&= \Prob{\eta,\nu_{t',i'}}{H(D)<\eta \text{ and } \xi_{t'}\cdot(f_{t',i'}(D_{t'}) + \nu_{t',i'}-T) \geq \eta} \\
&= \Prob{\eta,\nu_{t',i'}}{\eta \in (H(D),\xi_{t'}\cdot(f_{t',i'}(D_{t'}) + \nu_{t',i'}-T)   } \\
&= \int_{-\infty}^\infty \int_{-\infty}^\infty \Prob{\nu_{t',i'}}{\nu_{t',i'}=v} \\ & \quad \cdot \Prob{\eta}{\eta = \eta_0}\ones[\eta_0 \in (H(D),\xi_{t'}\cdot(f_{t',i'}(D_{t'}) + v-T)] \ dv \ d\eta_0 \\
&:= *
\end{align*}
Apply a change of variables as follows:
\begin{align*}
\hat{v} &= v+\frac{H(D)-H(D')}{\xi_{t'}}+f_{t',i'}(D_{t'})-f_{t',i'}(D_{t'}') \\
\hat{\eta}_0 &= \eta_0+H(D)-H(D')
\end{align*}
We have \(|H(D)-H(D')|\leq \max_{t\geq t_0} \xi_t\Delta_t =\xi_{t_0}\Delta_{t_0}\) (by \(\xi_t\Delta_t\) being non-increasing),  and \(| f_{t',i'}(D_{t'})-f_{t',i'}(D_{t'}')|\leq \Delta_{t'}\). Therefore,

\begin{displaymath}
|\hat{v}-v|\leq \frac{\xi_{t_0}\Delta_{t_0}}{\xi_{t'}}+ \Delta_{t'}, |\hat{\eta}_0-\eta_0|\leq \xi_{t_0}\Delta_{t_0}
\end{displaymath}

Apply this change of variable to get

\begin{align*}
* &= \int_{-\infty}^\infty \int_{-\infty}^\infty \Prob{\nu_{t',i'}}{\nu_{t',i'}=\hat{v}}\Prob{\eta}{\eta = \hat{\eta}_0}\ones[\eta_0+H(D)-H(D')\in \\ &(H(D),\xi_{t'}\cdot(f_{t',i'}(D_{t'}) + \hat{v}-T)] \ dv \ d\eta_0 \\
&=\int_{-\infty}^\infty \int_{-\infty}^\infty \Prob{\nu_{t',i'}}{\nu_{t',i'}=\hat{v}}\Prob{\eta}{\eta = \hat{\eta}_0}\ones[\eta_0\in \\ &(H(D'),\xi_{t'}\cdot(f_{t',i'}(D_{t'}) + \hat{v}-T)+H(D')-H(D)] \ dv \ d\eta_0 \\
&= \int_{-\infty}^\infty \int_{-\infty}^\infty \Prob{\nu_{t',i'}}{\nu_{t',i'}=\hat{v}}\Prob{\eta}{\eta = \hat{\eta}_0}\ones[\eta_0\in \\ &(H(D'),\xi_{t'}\cdot(f_{t',i'}(D_{t'}') + v-T)] \ dv \ d\eta_0 \\
&\leq \int_{-\infty}^\infty \int_{-\infty}^\infty \exp\left(\frac{\xi_{t'}}{4}\left(\frac{\xi_{t_0}\Delta_{t_0}}{\xi_{t'}}+ \Delta_{t'}\right)\right)\Prob{\nu_{t',i'}}{\nu_{t',i'}=v}\exp\left(\frac{\xi_{t_0}\Delta_{t_0}}{2}\right)\Prob{\eta}{\eta = \eta_0}\\&\quad \cdot\ones[\eta_0\in(H(D'),\xi_{t'}\cdot(f_{t',i'}(D_{t'}') + v-T)] \ dv \ d\eta_0 \\
&=\exp\left(\frac{\xi_{t'}}{4}\left(\frac{\xi_{t_0}\Delta_{t_0}}{\xi_{t'}}+ \Delta_{t'}\right)+\frac{\xi_{t_0}\Delta_{t_0}}{2}\right) \\
&\quad \cdot \Prob{\eta,\nu_{t',i'}}{H(D')<\eta \text{ and } \xi_{t'}\cdot(f_{t',i'}(D_{t'}') + \nu_{t',i'}-T) \geq \eta} \\
&=\exp\left(\frac{\xi_{t_0}\Delta_{t_0}}{4}+\frac{\xi_{t'}\Delta_{t'}}{4}+\frac{\xi_{t_0}\Delta_{t_0}}{2}\right)\Prob{\eta,\nu_{t',i'}}{A'=a} \\
&\leq\exp\left(\xi_{t_0}\Delta_{t_0}\right)\Prob{\eta,\nu_{t',i'}}{A'=a}   
\end{align*}

The first inequality comes from the bounds on \(|\hat{v}-v|,|\hat{\eta}_0-\eta_0|\) and the pdf of Laplace distribution. The last inequality is by \(\xi_t\Delta_t\) being non-increasing.
\end{proof}

Next, we define accuracy for threshold algorithms and then give an accuracy guarantee for ATG.
\begin{definition}[Accuracy of Threshold Answers]
\label{SAT-accuracy-def} For $\alpha,\beta>0$, an algorithm that produces threshold answers is \((\alpha,\beta\))-accurate for (finite) query stream $F=\set{\set{f_{t,i}}_{i=1}^{\ell_t}}_{t_0\le t\le t'}$ and threshold $T$ if for any input database stream $D=\set{D_t}_{t\ge t_0}$, with probability at least \(1-\beta\) the algorithm outputs answers $a_{t,i}\in\set{\bot,\top}$ such that it does not halt before receiving the final query $f_{t',\ell_{t'}}$ and  for all  \(a_{t,i}=\top\),
\begin{displaymath}
 f_{t,i}(D_t)\geq T-\alpha
\end{displaymath}
and for all \(a_{t,i}=\bot\),
\begin{displaymath}
f_{t,i}(D_t)\leq T+\alpha
\end{displaymath}
\end{definition}


\begin{theorem}[Accurary of ATG]
\label{SAT-accurary-theorem} Let \(\xi,\Delta:\{t_0,t_0+1,\ldots\}\rightarrow \R^+\) be such that both \(\Delta,\xi \cdot  \Delta\) are non-increasing and \(\xi\) is non-decreasing. For $\alpha>0$, ATG$(D,F,T,\xi)$ is $(\alpha,\beta)$-accurate for threshold $T$ and finite query stream $F=\set{\set{f_{t,i}}_{i=1}^{\ell_t}}_{t_0\le t\le t'}$ with sensitivity \(\Delta_{f_{t,i}}\leq \Delta_t\) for all \((t,i)\) and  $f_{t,i}(D_t)<T-\alpha$ for all $(t,i)\ne (t',\ell_{t'})$ and $f_{t',\ell_{t'}}(D_{t'})\ge T+\alpha$ as long as  
$$
\beta \ge \sum_{t=t_0}^{t'} \ell_{t}\exp(-\frac{\alpha \xi_t}{8})+\exp(-\frac{\alpha \xi_{t_0}}{8})
$$
\end{theorem}
\begin{proof}
First, we want to show that \(a_{t,i}=\bot\) if and only if \((t,i)<(t',i')\) with high probability.
This is true if we can show that for all \((t,i)\),
\begin{equation}
|\nu_{t,i}-\frac{\eta}{\xi_t} |\leq \alpha \label{eq:ATG-acc-requirement}
\end{equation}
Because if so, we have that for all \((t,i)<(t',i')\), \(f_{t,i}(D_t)+\nu_{t,i}< (T-\alpha)+(\alpha+\frac{\eta}{\xi_t}  )=\hat{T}_t\), so ATG will output \(\bot\), and that \(f_{t',i'}(D_{t'})+\nu_{t',i'}\geq(T+\alpha)+(\frac{\eta}{\xi_t}-\alpha)=\hat{T}_t \), so ATG will output \(a_{t',i'}=\top  \).
To show \eqref{eq:ATG-acc-requirement}, it is sufficient to require \(|\nu_{t,i}|\leq\alpha/2\) and \(|\frac{\eta}{\xi_t}|\leq \alpha/2\) for all \((t,i)\). The first requirement is false with probability \(\exp(-\frac{\alpha}{2}\frac{\xi_t}{4})=\exp(-\frac{\alpha \xi_t}{8})\) for each \((t,i)\). The second requirement is equivalent to \(|\frac{\eta}{\xi_{t_0}}|\leq\alpha/2\) for a single time step \(t_0\), which is false with probability \(\exp(-\frac{\alpha \xi_{t_0}}{4})\).
By union bound, \eqref{eq:ATG-acc-requirement} is true except probability at most \( \exp(-\frac{\alpha \xi_{t_0}}{4})+\sum_{t=t_0}^{t'} \ell_{t}\exp(-\frac{\alpha \xi_t}{8})\leq\exp(-\frac{\alpha \xi_{t_0}}{8})+ \sum_{t=t_0}^{t'} \ell_{t}\exp(-\frac{\alpha \xi_t}{8})\).
\cut{Therefore,\begin{displaymath}
\beta \leq\exp(-\frac{\alpha \xi_{t_0}}{4})+\sum_{t=t_0}^{t'} \ell_{t}\exp(-\frac{\alpha \xi_t}{8})\leq \sum_{t=t_0}^{t'} \ell_{t}\exp(-\frac{\alpha \xi_t}{8})+\exp(-\frac{\alpha \xi_{t_0}}{8})
\end{displaymath}}
\end{proof}

\subsection{Numeric Above Threshold for Growing Databases}
Next we analyze privacy and accuracy for NATG, which extends ATG by outputting a noisy answer to the first above threshold query.

\begin{algorithm} \label{SNAT-alg-code}
\caption{\textsc{NATG}($D,F,T,\xi$)} 
\begin{algorithmic}
\For {each incoming query \(f_{t,i}\)}
\If {$i=1$}
        \State Let \(\hat{T}_t\la T+\Lap(\frac{2}{\xi_t})\) \Comment Noisy threshold for queries at time \(t\)
        \EndIf
\State Let \(\nu_{t,i}\la \Lap(\frac{4}{\xi_t})\) \Comment{Noise for query $f_{t,i}$}

\If {$f_{t,i}(D_t)+\nu_{t,i}\geq \hat{T}_t$}
    \State Output $a_{t,i} \la f_{t,i}(D_t)+\Lap(\frac{8}{\xi_t})$ and halt \Comment{Abort on first query above threshold}
\Else
    \State Output \(a_{t,i}\la \bot\)
\EndIf

\EndFor

\end{algorithmic}
\end{algorithm}

\begin{theorem}[Privacy of NATG]
\label{SNAT-privacy-theorem}
Let \(\xi,\Delta:\{t_0,t_0+1,\ldots\}\rightarrow \R^+\) be such that both \(\Delta,\xi \cdot  \Delta\) are non-increasing and \(\xi\) is non-decreasing. Let $F=\{\{f_{t,i}\}_{i=1}^{\ell_t}\}_{t\ge t_0}$ be a stream of queries with sensitivity \(\Delta_{f_{t,i}}\leq \Delta_t\) for all \((t,i)\). Then for any  \(a\) in the range of ATG halting at time $t'$ and any neighboring database streams \(D,D'\) of starting size $t_0$,

\begin{displaymath}
\Prob{}{\text{NATG}(D,F,T,\xi)= a}\leq \exp\left(\xi_{t_0}\Delta_{t_0}+\frac{\xi_{t'}\Delta_{t'}}{8}\right)\Prob{}{\text{NATG}(D',F,T,\xi)=a}
\end{displaymath}\end{theorem}
\begin{proof}
NATG privacy loss is the sum of ATG privacy loss and the loss by Laplace noise \(\Lap(\frac{8}{\xi_{t'}})\) added to the numeric answer \(a_{t',i'}\). The first is \(\Delta_{t_0}\xi_{t_0}\) by Theorem \ref{SAT-privacy-theorem} and the latter is \(\xi_{t'}\Delta_{t'}/8\) because \(f_{t',i}\) has sensitivity at most \(1/\Delta_{t'}\). 
\end{proof}

Next, we define accuracy for numeric threshold algorithms and  give an accuracy guarantee for NATG.
\begin{definition}[Accuracy of Threshold and Numeric Answers]
\label{SNAT-accuracy-def}
For $\alpha,\beta>0$, an algorithm that produces threshold and numeric answers is $(\alpha,\beta)$-accurate for (finite) query stream $F=\set{\set{f_{t,i}}_{i=1}^{\ell_t}}_{t_0\le t\le t'}$ and threshold $T$ if for any input database stream $D=\set{D_t}_{t\ge t_0}$, with probability $1-\beta$ the algorithm outputs answers \(a_{t,i}\in \R\cup {\bot}\) such that it does not halt before receiving the final query $f_{t',\ell_{t'}}$ and  for all  \(a_{t,i}\in\R\),
\begin{displaymath}
|f_{t,i}(D_t)-a_{t,i}|\leq \alpha \text{ and } f_{t,i}(D_t)\geq T-\alpha
\end{displaymath}
and for all \(a_{t,i}=\bot\),
\begin{displaymath}
f_{t,i}(D_t)\leq T+\alpha
\end{displaymath}
\end{definition}

\begin{theorem}[Accurary of NATG]
\label{SNAT-accurary-theorem}

Let \(\xi,\Delta:\{t_0,t_0+1,\ldots\}\rightarrow \R^+\) be such that both \(\Delta,\xi \cdot  \Delta\) are non-increasing and \(\xi\) is non-decreasing. For $\alpha>0$, NATG$(D,F,T,\xi)$ is $(\alpha,\beta)$-accurate for threshold $T$ and finite query stream $F=\set{\set{f_{t,i}}_{i=1}^{\ell_t}}_{t_0\le t\le t'}$ with sensitivity \(\Delta_{f_{t,i}}\leq \Delta_t\) for all \((t,i)\) and  $f_{t,i}(D_t)<T-\alpha$ for all $(t,i)\ne (t',\ell_{t'})$ and $f_{t',\ell_{t'}}(D_{t'})\ge T+\alpha$ as long as  
$$
\beta \ge \sum_{t=t_0}^{t'} \ell_{t}\exp(-\frac{\alpha \xi_t}{8})+\exp(-\frac{\alpha \xi_{t_0}}{8})+\exp(-\frac{\alpha \xi_{t'}}{8})
$$\end{theorem}
\begin{proof}
We apply Theorem \ref{SAT-accurary-theorem} so that except with at most probability \( \sum_{t=t_0}^{t'} \ell_{t}\exp(-\frac{\alpha \xi_{t}}{8})+\exp(-\frac{\alpha \xi_{t_0}}{8})\), NATG is accurate for threshold answers, i.e. \(f_{t,i}(D_t)\geq T-\alpha\) and \(f_{t,i}(D_t)\leq T+\alpha\) hold for \(a_{t,i}\in \R\) and \(a_{t,i}=\bot\), respectively.

It's left to show that \(|f_{t',i'}(D_{t'})-a_{t',i'}|\leq \alpha\). This is true except with probability \(\exp(-\frac{\alpha \xi_{t'}}{8}) \) by Laplace mechanism. Therefore,

\begin{displaymath}
\beta \leq\sum_{t=t_0}^{t'} \ell_{t}\exp(-\frac{\alpha \xi_{t}}{8})+\exp(-\frac{\alpha \xi_{t_0}}{8})+\exp(-\frac{\alpha \xi_{t'}}{8})
\end{displaymath}\end{proof}

\subsection{Numeric Sparse for Growing Databases}\label{sec:nsg}
We now compose NATG multiple times into numeric sparse for growing databases (NSG).
Note that we may run NATG infinitely many times as long as there is input coming online. Any query that causes NATG to output a number and halt is called hard; any other query is called easy. 
\begin{algorithm} \label{SNS-alg-code}
\caption{\textsc{NSG}($D,F,T,\xi$) 
}
\begin{algorithmic}
\For {each incoming query $f_{t,i}$}
\If{no NATG subroutine is currently running}
\State Initialize a NATG subroutine with the same arguments
\EndIf
\State Output the NATG subroutine's output for $f_{t,i}$.
\EndFor
\end{algorithmic}
\end{algorithm}

Analysis of  privacy of NSG can be done by simply summing up the privacy loss stated for NATG. Since privacy for NATG changes over time, this privacy loss is then dependent on the time that NATG starts and ends, i.e. when hard queries come.  

\begin{theorem}[Privacy of NSG]\label{thm.nsgpriv}
\label{SNS-privacy-theorem}Let \(\xi,\Delta:\{t_0,t_0+1,\ldots\}\rightarrow \R^+\) be such that both \(\Delta,\xi \cdot  \Delta\) are non-increasing and \(\xi\) is non-decreasing. Let $F=\{\{f_{t,i}\}_{i=1}^{\ell_t}\}_{t\ge t_0}$ be a stream of queries with sensitivity \(\Delta_{f_{t,i}}\leq \Delta_t\) for all \((t,i)\). Then for any $a$ in the range of NSG with \(h_t\) hard queries arriving at time \(t\) for each \(t\geq t_0\) and any neighboring database streams $D,D'$ of starting size $t_0$,
\begin{displaymath}
\Prob{}{\text{NSG}(D,F,T,\xi)= a}\leq \exp\left(\xi_{t_0}\Delta_{t_0}+\frac{9}{8}\sum_{t\ge t+0} h_t\xi_{t}\Delta_t\right)\Prob{}{\text{NSG}(D',F,T,\xi)=a}
\end{displaymath}
\end{theorem}
\begin{proof}
Let \(t_0^j,t_1^j\) be the start and end time of \(j\)th round of NATG in NSG. Then by Theorem \ref{SNAT-privacy-theorem},
the privacy loss is
at most
\begin{align*}
\epsilon &:=\sum_{j}  \left( \xi_{t_0^j} \Delta_{t_0^j}+\frac{\xi_{t_1^j}\Delta_{t_1^j}}{8} \right)
\end{align*}
The start time \(t_0^j\) of round \(j\) is at least the end time of last round \(t_1^{j-1}\) for each \(j\geq 2\), and \(t_0^1 \geq n\), so
\begin{align*}
\epsilon &\leq \xi_n\Delta_{n} + \sum_{j}\left(\xi_{t_1^j}\Delta_{t_1^j}+\frac{\xi_{t_1^j}\Delta_{t_1^j}}{8} \right) \\
&= \xi_n\Delta_n +\frac{9}{8}\sum_{t=n}^\infty h_t\xi_{t}\Delta_t
\end{align*}
The last equality is by the fact that the end times of NATG are exactly when hard queries come.

\end{proof}

Now we apply Theorem \ref{SNS-privacy-theorem} to a more specific setting of linear queries.
\begin{corollary}
\label{SNS-privacy-corollary} For adaptively chosen linear query stream $F$, NSG$(D,F,T,\xi)$ for database streams of starting size $n$ and noise function \(\xi_{t}=t^p\) for some \(p\in [0,1]\) is \((n^{p-1} +\frac{9}{8}\sum_{t=n}^\infty h_tt^{p-1},0)\)-DP.
\end{corollary}

\begin{proof}
From the fact that \(\xi_t\) is non-decreasing, that linear queries have sensitivity \(\Delta_t=1/t,\) and that \(\xi_{t}\Delta_t=t^{p-1}\) is non-increasing, the conditions satisfy Theorem \ref{SAT-privacy-theorem}'s assumption, so we can apply Theorem \ref{SNS-privacy-theorem}.\end{proof}Corollary \ref{SNS-privacy-corollary} suggests that in order to bound privacy loss, we need to bound the number of hard queries, especially those arriving early in time. 

Now we state the accuracy for NSG.
Note that we use the same Definition \ref{SNAT-accuracy-def} from NATG, since NSG also produces threshold and numeric answers. 
\begin{theorem}[Accuracy of NSG]
\label{SNS-accurary-theorem}
Let \(\xi,\Delta:\{t_0,t_0+1,\ldots\}\rightarrow \R^+\) be such that both \(\Delta,\xi \cdot  \Delta\) are non-increasing and \(\xi\) is non-decreasing. Let $k:\set{t_0,t_0+1,\dots}\to \R$ be a query budget. For $\alpha>0$, NATG$(D,F,T,\xi)$ is $(\alpha,\beta)$-accurate for threshold $T$ and finite query stream $F=\set{\set{f_{t,i}}_{i=1}^{\ell_t}}_{t_0\le t\le t'}$ with sensitivity \(\Delta_{f_{t,i}}\leq \Delta_t\) for all \((t,i)\) and respecting cumulative query budget for $\sum_{\tau=t_0}^t \ell_t\le \sum_{\tau=t_0}^t k_t$ for all $t\ge t_0$ as long as  
\begin{align}
\beta &\ge \exp(-\frac{\alpha \xi_n}{8}) + \sum_{t=n}^\infty (\ell_{t}+2h_t)\exp(-\frac{\alpha \xi_{t}}{8}), \label{SNS-beta-value-1} 
\end{align}
where $h_t=\abs{\set{i:f_{t,i}(D_t)\ge T-\alpha}}$ for $t\ge t_0$. This implies $(\alpha,\beta)$-accuracy for  
\begin{align}
\beta &\ge \exp(-\frac{\alpha \xi_n}{8}) +3  \sum_{t=n}^\infty k_t\exp(-\frac{\alpha \xi_{t}}{8}). \label{SNS-beta-value-2}
\end{align}

\end{theorem}
\begin{proof} 
We need to show that except with probability at most \(\beta\): 
\begin{enumerate}
\item 
For each \(a_{t,i}=\bot\), \(f_{t,i}(D_t) \leq T+\alpha\)
\item
For each \(a_{t,i}\in \R\), \(f_{t,i}(D_t) \geq T-\alpha\)
\item
For each \(a_{t,i}\in \R\), \(|f_{t,i}(D_t) -a_{t,i}|\leq \alpha\)

\end{enumerate}
Suppose the \(j\)th round of NATG starts and ends at time \(t_0^j,t_1^j\), and answers  \(\ell_{t}^j\)  queries at time \(t\). The set of three conditions is equivalent to  requiring that \(j\)th NATG\ round satisfies Definition \ref{SNAT-accuracy-def} with respect to threshold \(T \)  and stream of queries that \(j\)th round of NATG\ answers. By Theorem \ref{SNAT-accurary-theorem} and union bound, the rest of the proof is a calculation:\ we can take \(\beta\) to be
\begin{align*}
\beta &= \sum_j \left(\sum_{t=n}^\infty  \ell_{t}^j\exp(-\frac{\alpha \xi_{t}}{8})+\exp(-\frac{\alpha\xi_{t_0^j}}{8})+\exp(-\frac{\alpha\xi_{t_1^j}}{8}) \right) \\
&\leq \sum_j \left(\sum_{t=n}^\infty  k^j_t\exp(-\frac{\alpha \xi_{t}}{8})\right)+\sum_j\left(\exp(-\frac{\alpha\xi_{t_1^{j-1}}}{8})+\exp(-\frac{\alpha\xi_{t_1^{j-1}}}{8})\right) \\
&\leq \left(\sum_{t=n}^\infty  \ell_{t}\exp(-\frac{\alpha \xi_{t}}{8})\right) +\exp(-\frac{\alpha \xi_n}{8}) + \sum_j 2\exp(-\frac{\alpha \xi_{t_1^j}}{8})
\end{align*}
The first inequality is by the fact that start time of the next round is after the end of the current round, and we let \(t_1^0=n\) for convenience. By condition (2), NATG can only halt on queries \(f_{t,i}\) such that \(f_{t,i}(D_t) \geq T-\alpha\), so there are at most \(h_t\) rounds of NATG halting at time \(t\). Therefore, 
\begin{align*}
2\sum_j \exp(-\frac{\alpha \xi_{t_1^j}}{8}) \leq2\sum_{t=n}^\infty h_t\exp(-\frac{\alpha \xi_{t}}{8}) 
\end{align*}
which finishes the proof for \eqref{SNS-beta-value-1}. Inequality \eqref{SNS-beta-value-2} follows from the observations that \(\ell_{t}\geq h_t\) 
 and that \(\exp(-\frac{\alpha \xi_{t}}{8})\) is a non-increasing function of \(t\).
\end{proof}


\cut{
Here we give some example of what SNS accuracy may look like for our setting of linear queries and a natural choice of query budget. 
\begin{corollary}
For two neighboring sequences of growing databases \(D\sim D'\) and adaptively chosen linear queries \(f_{t,i}\),  SNS which starts at time \(n\), with noise function \(\xi_{t}=t^p\) for some \(p\in [0,1]\) and query budget \(k_t=q_0\exp(\frac{\alpha t^p}{16})\) for some constant \(q_0\) such that \(q(n)\geq 1\) is \(\epsilon\)-DP and \(\alpha,\beta\)-accurate for
\begin{align*}
\beta = 
\end{align*}
 
\end{corollary}
\begin{proof}
By Corollary \ref{SNS-privacy-corollary}, SNS is \(\eps'\)-DP for \(\epsilon' = n^{p-1} +\frac{9}{8}\sum_{t=n}^\infty h_tt^{p-1}\). 
\end{proof}
}